\documentclass[11pt]{article}
\usepackage{times}
\usepackage{helvet}
\usepackage{courier}
\usepackage{graphicx}
\usepackage{wrapfig}
\usepackage{dsfont}
\usepackage{nicefrac}
\usepackage{booktabs}
\usepackage{mathtools}
\usepackage{comment}
\usepackage[vlined,ruled,linesnumbered]{algorithm2e}
\usepackage{paralist}
\usepackage[colorlinks,citecolor=green!60!black,linkcolor=red!60!black]{hyperref}
\usepackage[sort&compress,nameinlink]{cleveref}

\usepackage{amsmath}
\usepackage{amsthm}
\usepackage{amssymb}

\usepackage{todonotes}

\renewcommand{\sc}{{{\mathrm{sc}}}}

\newcommand{\score}{{{\mathrm{score}}}}

\newcommand{\naturals}{{{\mathbb{N}}}}

\newcommand{\reals}{{{\mathbb{R}}}}

\newtheorem{theorem}{Theorem}

\newtheorem{corollary}{Corollary}
\newtheorem{lemma}{Lemma}

\newtheorem{definition}{Definition}
\newtheorem{example}{Example}
\newtheorem{proposition}{Proposition}

\crefname{table}{Table}{Tables}
\Crefname{table}{Table}{Tables}
\crefname{figure}{Figure}{Figures}
\crefname{theorem}{Theorem}{Theorems}
\crefname{definition}{Definition}{Definitions}
\crefname{corollary}{Corollary}{Corollaries}
\crefname{observation}{Observation}{Observations}
\crefname{lemma}{Lemma}{Lemmas}
\crefname{example}{Example}{Examples}
\crefname{reduction}{Reduction}{Reductions}
\crefname{construction}{Construction}{Constructions}
\crefname{subsection}{Subsection}{Subsections}
\crefname{section}{Section}{Sections}
\crefname{proposition}{Proposition}{Propositions}
\crefname{algorithm}{Algorithm}{Algorithms}
\Crefname{equation}{Inequality}{Inequalities}

\newcommand{\calR}{\mathcal{R}}

\newcommand{\calA}{\mathcal{A}}

\crefname{eq}{equality}{equalities}

\newcommand{\pav}{{{{\mathrm{PAV}}}}}
\newcommand{\phrag}{{{{\mathrm{Phrag}}}}}

\newcommand{\hard}{{{{\mathrm{hard}}}}}
\newcommand{\last}{{{{\mathrm{last}}}}}

\newcommand{\seqpav}{{{{\mathrm{seq\text{-}PAV}}}}}

\newcommand{\opt}{{{{\mathrm{opt}}}}}

\def\argmax{\mbox{argmax}}

\newcommand{\shortcite}[1]{\cite{#1}}

\makeatletter
\newtheorem*{rep@theorem}{\rep@title}
\newcommand{\newreptheorem}[2]{%
\newenvironment{rep#1}[1]{%
 \def\rep@title{#2 \ref{##1}}%
 \begin{rep@theorem}}%
 {\end{rep@theorem}}}
\makeatother
\newreptheorem{theorem}{Theorem}
\newreptheorem{proposition}{Proposition}
\newreptheorem{lemma}{Lemma}
\newreptheorem{corollary}{Corollary}

\title{Proportionality Degree of Multiwinner Rules}

\author{Piotr Skowron\\
University of Warsaw\\
\texttt{p.skowron@mimuw.edu.pl}}

\date{}

\allowdisplaybreaks

\begin{document}

\maketitle

\begin{abstract}
We study multiwinner elections with approval-based preferences. An instance of a multiwinner election consists of a set of alternatives, a population of voters---each voter approves a subset of alternatives, and the desired committee size $k$; the goal is to select a committee (a~subset) of $k$ alternatives according to the preferences of the voters. We investigate a number of election rules and ask whether the committees that they return represent the voters proportionally. In contrast to the classic literature, we employ quantitative techniques that allow to measure the extent to which the considered rules are proportional. This allows us to arrange the rules in a clear hierarchy. For example, we find that Proportional Approval Voting (PAV) has better proportionality guarantees than its sequential counterpart, and that Phragm\'{e}n's Sequential Rule is worse than Sequential PAV. Yet, the loss of proportionality for the two sequential rules is moderate and in some contexts can be outweighed by their other appealing properties. Finally, we measure the tradeoff between proportionality and utilitarian efficiency for a broad subclass of committee election rules.      
\end{abstract}

\section{Introduction}

An approval-based committee election rule (an ABC rule, in short) is a function that given a set of~$m$ candidates~$C$ (the candidates are also referred to as alternatives), a population of~$n$ voters---each approving a subset of~$C$---and an integer~$k$ representing the desired committee size, returns a $k$-element subset of~$C$. Committee election rules are important tools that facilitate collective decision making in various contexts such as electing representative bodies (e.g., supervisory boards, trade unions, etc.), finding responses to database querying~\cite{dwo-kum-nao-siv:c:rank-aggregation, proprank}, suggesting collective recommendations for groups~\cite{budgetSocialChoice, bou-lu:c:value-directed}, and managing discussions on proposals within liquid democracy~\cite{BKNS14a}. Further, since committee elections are a special case of participatory budgeting (PB)~\cite{participatoryBudgeting, knapsackVoting}, good understanding of ABC rules is a prerequisite for designing effective PB methods.      

The applicability of various ABC rules often depends on the particular context, yet an important requirement one often imposes on a committee election rule is that it should be fair to (groups of) voters. While fairness is a broad concept putting various principles under the same umbrella, in certain types of collective decision-making---specifically when the goal is to elect a committee---it is often argued that a fair rule should be \emph{proportional}\footnote{Other basic axioms describing fairness requirements are, e.g., anonymity or solid coalitions property~\cite{dum:b:voting}.}, as illustrated by the following example.\footnote{Indeed, proportional representation (PR) electoral systems are often argued to be more fair than non-proportional ones~\cite{grofmanProspectives,grofmanChoosingElectoral}. Yet, the criticism of proportionality also appears in the literature, and it comes from the two main directions---one stems from the analysis of voting power in the elected committees~\cite{FesMac98, shapley_shubik_1954}; the other one is grounded in the arguments in favor of degressive proportionality~\cite{RePEc:ucp:jpolec:doi:10.1086/670380}. Proportionality is also related to fairness in other domains, such as allocation of individual~\cite{fair_division_steinhaus} and public goods~\cite{FMS18}.} 
\begin{example}\label{ex:simple_example}
Consider an election with 30 candidates, $c_1, c_2, \ldots, c_{30}$, and 100 voters having the following approval-based preferences. The first 60 voters approve the subset $C_1 = \{c_1, \ldots, c_{10}\}$; the next 30 voters approve $C_2 = \{c_{11}, \ldots, c_{20}\}$, and the last 10 approve $C_3 = \{c_{21}, \ldots, c_{30}\}$. When the size of the committee to be elected is 10, then a proportional rule should select a committee with $6$ members coming from $C_1$, 3 members coming from $C_2$, and 1 member from $C_3$.
\end{example}

Identifying a proportional committee in \Cref{ex:simple_example} was easy due to the voters' preferences having a very specific structure: each two approval sets were either the same or disjoint. In a more general case, when the approval sets can arbitrarily overlap, the answer is no longer straightforward. Several approaches to formalizing proportionality have been proposed in the literature~\cite{justifiedRepresenattion, AEHLSS18, bri-las-sko:c:apportionment, FMS18, lac-sko:t:approval-thiele, pjr17}. These approaches are mainly axiomatic---they formally define natural properties referring to proportionality, and classify known rules based on whether they satisfy these axioms or not. With a few notable exceptions (e.g., \cite{FMS18,pjr17}) the approaches from the literature are qualitative, giving only a yes/no answer to the question of ``Is a given rule proportional or not?''. Hence, in essence, they are not capable of measuring the extent to which proportionality is satisfied or violated. This is a serious drawback since there is no single rule satisfying all the desired properties, and the choice of the rule boils down to a judgment call. The mechanism designer must decide on which properties she finds most critical and which tradeoff she is willing to accept. However, in order to make such a decision one primarily needs to understand these tradeoffs; in particular the extent to which properties of interest are violated.    

In our study we employ a quantitative approach. For several rules we will determine their \emph{proportionality degrees}, i.e., we will assess the extent to which they satisfy proportionality. Informally speaking, the proportionality degree of a rule is a function specifying how the rule treats groups of voters with cohesive preferences, depending on the size of these groups. At the same time, our proportionality degree has the form of a guarantee---it gives the best possible bounds on the proportionality that the rule cannot violate, no matter what are the voters' preferences.   

Our definition of proportionality degree is closely related to the concept proposed by Skowron~et~al.~\shortcite{proprank}. In fact, their work establishes certain bounds on the proportionality degree for a number of ABC rules. Unfortunately, these bounds cannot be used for comparing the rules, since they are not tight; not even asymptotically. Our main contribution is that we derive new almost tight bounds that allow to arrange the rules that we study in a clear hierarchy, based on how proportional they are. A tight estimation of the proportionality degree is already known for Proportional Approval Voting (PAV)~\cite{AEHLSS18}---we generalize this result and calculate the proportionality degree for convex Thiele methods, a broad class of rules which---in particular---includes PAV. Additionally, we find close estimations of the proportionality degree for two sequential methods often considered in the literature: Sequential PAV and Phragm\'{e}n's Sequential Rule.

Our findings can be summarized as follows. As far as the proportionality guarantees are concerned, Sequential PAV is better than Phragm\'{e}n's Sequential Rule; further, PAV is better than both of the sequential rules. For reasonable committee sizes the proportionality degree of Sequential PAV is no worse than roughly 70\% of the proportionality degree of PAV. The proportionality degree of Phragm\'{e}n's Sequential Rule is roughly half as much as that for PAV. On the one hand, our results suggest that PAV should be preferred whenever proportionality is the primary goal. On the other hand, they demonstrate that the loss of proportionality for Sequential PAV and Phragm\'{e}n's Sequential Rule is moderate. Thus, using these rules can be justified in cases when the decision maker considers their other distinctive properties equally important to proportionality (in \Cref{sec:phragmen,sec:seqpav} we briefly recall a few arguments that can speak in favor of one of the two sequential rules).      
  
Finally, in \Cref{sec:utilitarian_efficiency} we apply the same quantitative methodology to criteria other than proportionality. Specifically, we focus on the utilitarian efficiency, measured as the total number of approvals that the members of the elected committee obtain. We establish accurate estimations on the loss of the utilitarian efficiency for convex Thiele methods. For Sequential PAV and for Phragm\'{e}n's Sequential Rule such estimations are already known~\cite{lac-sko:quantitative}. Together with the aforementioned proportionality guarantees, this allows to quantify the tradeoff between the level of proportionality and utilitarian efficiency, that is to estimate the price of fairness for particular (classes of) rules. Similar tradeoffs have been considered in the resource allocation domain~\cite{BertsimasFT11,CKKK12}.

\section{The Model}

For each natural number $p \in \naturals$, we set $[p]= \{1, 2, \ldots, p\}$ and $[p]_{0} = \{0, 1, 2, \ldots, p\}$.
For each set~$X$ by $S_p(X)$ we denote the set of all $p$-element subsets of $X$, and by $S(X)$ we denote the powerset of $X$, i.e., $S(X) = \bigcup_{p = 0}^{|X|} S_p(X)$.

An approval-based election is a triple $(N, C, A)$, where $N = \{1, 2, \ldots n\}$ is the set of \emph{voters}, $C = \{c_1, c_2, \ldots, c_m\}$ is the set of \emph{candidates}, and $A$ is an \emph{approval-based profile} (or, in short, a \emph{profile}), i.e., a function $A\colon N \to S(C)$ that maps each voter to a subset of $C$; intuitively, $A(i)$ consists of candidates that voter $i$ finds acceptable and is referred to as the \emph{approval set} of $i$. Conversely, for each candidate $c \in C$ by $N(c)$ we denote the set of voters who approve $c$, i.e., $N(c) = \{i \in N \mid c \in A(i)\}$. Whenever we consider an election, we will implicitly assume that $N$ and $C$ refer to the sets of voters and candidates, respectively. Similarly, we will always assume that $n$ and $m$ denote the number of voters and candidates, respectively. For the sake of the simplicity of the notation we will also identify elections with approval-based profiles, assuming that the sets of voters and candidates are implicitly encoded as, respectively, the domain and the union of elements in the range of $A$. We denote the set of all elections by $\calA$ (yet, following our convention, we will treat the elements of $\calA$ as if these were simply profiles).

\subsection{Approval-Based Committee Election Rules}\label{sec:rules_desc}

An approval-based committee election rule (an ABC rule, in short) is a function $\calR$ that for each approval-based profile $A$ and a positive integer $k$ returns a nonempty set of size-$k$ subsets of candidates, i.e., $\calR(A, k) \in S(S_k(C))$. We will refer to the elements of $S_k(C)$ as to size-$k$ \emph{committees}, or simply as to \emph{committees} when the committee size~$k$ is known from the context. We will call the elements of $\calR(A, k)$ \emph{winning committees}. We say that a candidate $c$ represents voter $i$ (or that $c$ is a representative of $i$) if $c$ is a member of a winning committee and if $c \in A(i)$.

Below, we recall definitions of several prominent (classes of) approval-based committee election rules, often studied in the literature in the context of proportional representation. 
\begin{description}
\item[Thiele Methods.] For a non-increasing function $\lambda\colon \reals \to \reals$ the $\lambda$-Thiele rule is defined as follows.\footnote{For the definition of the rules it is sufficient to have $\lambda$ specified only on the set of natural numbers, yet later on it will be sometimes convenient for us to work with the continuous version of $\lambda$.} The $\lambda$-score that a committee $W \in S_k(C)$ gets from a voter $i \in N$ is equal to $\sc_{\lambda}(W, i) = \sum_{j = 1}^{|A(i) \cap W|} \lambda(j)$. The total $\lambda$-score of $W$ is the sum of the scores that it garners from all the voters: $\sc_{\lambda}(W) = \sum_{i \in N} \sc_{\lambda}(W, i)$. The $\lambda$-Thiele rule returns the committees with the highest $\lambda$-score.

\item[Proportional Approval Voting (PAV).]  PAV is the $\lambda_\pav$-Thiele rule for $\lambda_\pav(i) = \nicefrac{1}{i}$. Using the harmonic sequence to compute the score ensures that PAV satisfies particularly appealing properties pertaining to proportionality~\cite{justifiedRepresenattion, bri-las-sko:c:apportionment, lac-sko:t:approval-thiele}. 

\item[Sequential Proportional Approval Voting (Seq-PAV).] This is an iterative rule that involves $k$ steps. It starts with an empty committee $W = \emptyset$ and in each step it adds to $W$ the candidate that increases the PAV score of $W$ most. Thus, intuitively, each voter starts with the same voting power (equal to 1), which can decrease over time. When the voter gets her $i$-th representative in the committee, then her voting power decreases from $\nicefrac{1}{i}$ to $\nicefrac{1}{i+1}$. 

\item[Phragm\'{e}n's Sequential Rule.] This is also an iterative rule that is usually defined as a load balancing procedure. Each candidate $c \in C$ is associated with one unit of load; if $c$ is selected this unit of load needs to be distributed among the voters who approve $c$. In each step the rule selects the candidate that minimizes the load assigned to the maximally loaded voter. Formally, let $\ell_i(j)$ be the total load assigned to voter $i$ after the $j$-th iteration ($\ell_i(0) = 0$ for all $i \in N$). In each step $j$ the rule selects a candidate $c$ and finds a distribution $\{\delta_{i}(j)\}_{i \in N}$ of one additional unit of load such that:
\begin{inparaenum}[(i)]
\item $\sum_{i \in N} \delta_{i}(j) = 1$,
\item $\delta_{i}(j) \geq 0$ for each $i \in N$, 
\item $\delta_{i}(j) = 0$ for each $i \notin N(c)$, and
\item $\max_{i \in N} \big(\ell_{i}(j-1) + \delta_{i}(j)\big)$ is minimized.
\end{inparaenum} 
(Note that minimizing the maximum load (iv) is the criterion for choosing both $c$ and $\delta$.)
Finally, for each $i \in N$ the rule updates the total load, $\ell_{i}(j) = \ell_{i}(j-1) + \delta_{i}(j)$. For more discussion on the rule and its properties we refer to the work of Brill~et~al.~\shortcite{aaai/BrillFJL17-phragmen} and to the survey by Janson~\shortcite{Janson16arxiv}.

\item[Phragm\'{e}n's Maximal Rule.] This method is similar to the Phragm\'{e}n's Sequential Rule, but in this case the rule does not proceed sequentially---the whole optimization is performed only once, globally. The criterion for choosing the winning committee and the corresponding load distribution is that the total load assigned to the maximally loaded voter should be minimized.
\end{description}

\subsection{Proportionality Degree of ABC Rules}

This section defines and explains the notion of proportionality used in our comparative study. 

Given an approval-based profile $A$ and a desired committee size $k \in \naturals$, we say that a group of voters $V \subseteq N$ is $\ell$-large if $|V| \geq \ell \cdot \frac{n}{k}$. The satisfaction of the group $V$ from a committee $W \in S_k(C)$ is defined as the average number of representatives that a voter from $V$ has in committee $W$:
\begin{align*}
\mathrm{sat}_V(A, W) = \frac{1}{|V|} \sum_{i \in V} |W \cap A(i)|\text{.}
\end{align*}

\begin{definition}\label{def:proportionality_guarantee}
We say that a function $g\colon \naturals \times \naturals \to \reals$ is a $k$-proportionality guarantee for an ABC rule $\calR$ if for each approval-based profile $A \in \calA$, each $\ell$-large group of voters $V \subseteq N$, and for each winning committee $W \in \calR(A, k)$, the following implication holds:
\begin{align*}
\Big|\bigcap_{i \in V} A(i)\Big| \geq g(\ell, k) \implies \mathrm{sat}_V(A, W) \geq g(\ell, k) \text{.}
\end{align*}
\end{definition}

Note that the condition in \Cref{def:proportionality_guarantee} must hold for all values of $n$ and $m\geq k$; in particular, $g$ cannot depend on $n$, $m$ nor on the specific structural properties of approval profiles.

Let us give an informal explanation of \Cref{def:proportionality_guarantee}; intuitively, this definition says that an $\ell$-large group of voters is guaranteed $g(\ell, k)$ representatives in the size-$k$ committee elected by a rule $\calR$, no matter what the preferences of the other voters are. The average satisfaction of at least $g(\ell, k)$ is guaranteed only to $\ell$-large groups that have coherent enough preferences, that is for groups that agree on some $g(\ell, k)$ common candidates. Indeed, it is clear that if each member of a group $V$ approves different candidates, then $V$ cannot be satisfied by any rule; in fact we will call such $V$ simply a \emph{set} instead of a \emph{group} to indicate that there is no agreement between the voters in $V$, so they cannot be related or grouped based of their preferences. Thus, summarizing, \Cref{def:proportionality_guarantee} specifies how the rule treats cohesive groups of certain size.   

One intuitively expects that a proportional rule should have a guarantee of $g(\ell, k) = \ell$. Indeed, a group $V$ such that
\begin{inparaenum}[(i)]
\item $|V| \geq \ell \cdot \frac{n}{k}$, and
\item $|\bigcap_{i \in V} A(i)| \geq \ell$,
\end{inparaenum}
is large enough to deserve $\ell$ representatives in the elected committee, and choosing $\ell$ candidates that all members of the group agree on is feasible. Unfortunately, such a guarantee is not possible to achieve. From the recent results of Aziz et al.~\shortcite{AEHLSS18} it follows that there exists no rule with the proportionality guarantee satisfying $g(\ell, k) > \ell-1$. The same work proves that PAV matches this lower bound.

\begin{theorem}[Aziz et al.~\shortcite{AEHLSS18}]\label{thm:pav_prop_degree}
PAV has the proportionality guarantee of $g(\ell, k) = \ell-1$.
\end{theorem}

In order to facilitate our further discussion we define two related concepts.

\begin{definition}
Let $\mathrm{guar}(\calR)$ be the set of all $k$-proportionality guarantees for a rule $\calR$.
The $k$-proportionality degree of an ABC rule $\calR$ is the function:
\begin{align*}
d_\calR(\ell, k) \; = \; \sup_{\mathclap{g \in \mathrm{guar}(\calR)}} \; g(\ell, k) \text{.}
\end{align*}
In other words, the $k$-proportionality degree is the best possible $k$-proportionality guarantee one can find for the rule. 
The proportionality degree of $\calR$ is the function $d_\calR(\ell) = \min_{k} d_\calR(\ell, k)$, i.e., it defines the best possible guarantee that holds irrespectively of the size of the committee.
\end{definition}

The idea of measuring proportionality of multiwinner rules as the average number of representatives that cohesive groups of voters get was first proposed by S{\'a}nchez-Fern{\'a}ndez et al.~\shortcite{pjr17} (they call the average number of representatives the \emph{average satisfaction} of a voter). Though, the first estimation of the proportionality degree for a multiwinner rules appears only in the subsequent work of Aziz et al.~\shortcite{AEHLSS18} (see~\Cref{thm:pav_prop_degree}, above); the analysis there uses a swap-optimality argument---they show that if there existed a committee $W$ witnessing that the proportionality guarantee of PAV is lower than $d_\pav(\ell) = \ell-1$, then it would be possible to find a pair of candidates, $c \in W$ and $c' \notin W$, such that replacing $c$ with $c'$ in $W$ would improve the PAV score of $W$. The same argument applies to a local-search heuristic for PAV yielding a polynomial-time rule with virtually the same proportionality degree as PAV---this result solved an open question from 2015. Next, Skowron~et~al.~\shortcite{proprank} analyzed the proportionality degree of a number of ABC rules (though, the definition of the proportionality degree appears there only implicitly), focusing on those that satisfy the committee enlargement monotonicity. However, their estimations are not asymptotically tight, and are not sufficient to compare the studied rules. E.g., for Phragm\'{e}n's Sequential Rule, they showed that $d_{\phrag}(\ell, k) \geq \frac{\ell(\ell+1)}{5k}$; we will strengthen this result by showing that $d_{\phrag}(\ell) = \frac{\ell}{2} \pm \frac{1}{2}$ (note that $k \geq \ell$, and that for ``small'' groups $k$ is much greater than $\ell$, and so the estimation given in~\cite{proprank} is highly inaccurate; specifically, it does not allow to obtain a positive bound on $d_{\phrag}(\ell)$ nor to compare Phragm\'{e}n's Sequential Rule with virtually any reasonable rule known in the literature; ).  
The aforementioned results from the literature together with our technical results are summarized in \Cref{tab:seq_pav_lower_bound2}.

\renewcommand{\arraystretch}{1.2} 
\begin{table}[tb!]
\begin{center}
   \begin{tabular*}{\linewidth}{@{\extracolsep{\fill}}@{}lp{4.5cm}p{5.3cm}@{}}
    \toprule
    \textbf{Rule} & \textbf{Proportionality degree} & \textbf{Utilitarian efficiency} \\ 
    \midrule
    seq-Phragm\'{e}n & $\frac{\ell -1}{2}$~(\Cref{thm:phrag_guarantee}) & $\Theta\left(\frac{1}{\sqrt{k}}\right)$~\cite{lac-sko:quantitative} \\ 
    \midrule
    max-Phragm\'{e}n & $1$~(\Cref{prop:max_phrag_guarantee}) & $\Theta\left(\frac{1}{k}\right)$~\cite{lac-sko:quantitative} \\ 
    \midrule
    seq-PAV & \mbox{$\approx 0.7\ell-1$ for $k \leq 200$ \mbox{(\Cref{sec:seqpav})}} & $\Theta\left(\frac{1}{\sqrt{k}}\right)$~\cite{lac-sko:quantitative} \\ 
    \midrule
    \textbf{Thiele Rules} & \multicolumn{2}{p{9.8cm}}{General characterizations are given in \Cref{thm:general_bounds,thm:utilitarian_efficiency} (lower bounds) and in \Cref{prop:general_bound_tight,prop:utilitarian_efficiency_tight} (upper bounds).} \\ 
    \midrule
    $\qquad\lambda(i) = \frac{1}{i}$ (PAV) & $\ell-1 + o\left(\frac{1}{k}\right)$~(\Cref{thm:general_bounds}, and~\cite{AEHLSS18}) & $\Theta\left(\frac{1}{\sqrt{k}}\right)$~(\Cref{thm:utilitarian_efficiency}, and \cite{lac-sko:quantitative})\\
    $\qquad\lambda(i) = \left(\frac{1}{i}\right)^{\nicefrac{2}{3}}$ & see~\Cref{fig:sqrt_pav_bounds_all} & $\Theta\left(\frac{1}{k^{\nicefrac{2}{5}}}\right)$~(\Cref{thm:utilitarian_efficiency}) \\  
    $\qquad\lambda(i) = \frac{1}{\sqrt{i}}$ & see~\Cref{fig:sqrt_pav_bounds_all} & $\Theta\left(\frac{1}{\sqrt[3]{k}}\right)$~(\Cref{thm:utilitarian_efficiency}) \\
    $\qquad\lambda(i) = \frac{1}{i^2}$ & see~\Cref{fig:sqrt_pav_bounds_all} & $\Theta\left(\frac{1}{k^{\nicefrac{2}{3}}}\right)$~(\Cref{thm:utilitarian_efficiency}) \\
    \bottomrule 
   \end{tabular*}
\end{center}
\caption{The proportionality degree and utilitarian efficiency for selected multiwinner rules.}
\label{tab:seq_pav_lower_bound2}
\end{table}

In the remainder of this section we compare the notion of the proportionality guarantee/degree with other concepts pertaining to proportionality, studied in the literature.
\begin{description}
\item[Extended Justified Representation (EJR).] EJR~\cite{justifiedRepresenattion} requires that each $\ell$-large group of voters $V$ with $|\bigcap_{i \in V} A(i)| \geq \ell$ must contain a voter who approves at least $\ell$ members of the winning committee(s). This property is very natural and interesting, yet it also has certain drawbacks. On the one hand, it seems quite weak---it aims at analyzing how voting rules treat \emph{groups} of voters, yet for each group it only enforces the requirement that there must exist \emph{some} voter in the group who is well represented. On the other hand, it appears very strong---if each voter from $V$ approves $\ell-1$ members of the winning committee, then $V$ already witnesses that the rule violates the property. Indeed, PAV is almost the only natural known rule satisfying EJR.\footnote{Other rules satisfying EJR are either very similar to PAV (e.g., the local search algorithm for PAV) or very technical and specifically tailored to satisfy the particular property~\cite{AEHLSS18}} Our approach, on the other hand, provides a fine-grained guarantee on how well a given rule represents certain groups of voters, and focuses on the average satisfaction of the voters within the group rather than on the satisfaction of the most satisfied voter.

Yet, there is also a relation between EJR and the proportionality degree. Clearly, if for some $\alpha \in [0,1]$, the proportionality guarantee of a rule $\calR$ satisfies \mbox{$g_\calR(\ell, k) \geq \alpha \ell$}, then---by the pigeonhole principle---$\calR$ also satisfies an $\alpha$-approximation of EJR. For the reverse direction a weaker implication has been shown by S{\'a}nchez-Fern{\'a}ndez~et~al.~\shortcite{pjr17}: if a rule $\calR$ satisfies EJR, then it has the proportionality guarantee of $\calR$ is \mbox{$g_\calR(\ell, k) = \frac{\ell-1}{2}$}.\footnote{The main argument bases on applying the EJR property recursively: Assume a rule satisfies EJR and consider an $\ell$-large group of voters $V$ with $|\bigcap_{i \in V} A(i)| \geq \ell$. Then by EJR, we know that there exists $v \in V$ with at least $\ell$ representatives. Now, observe that $V \setminus \{v\}$ is an $(\ell-1)$-large group and that it is cohesive. Thus, there exists $v' \in V \setminus \{v\}$ with at least $\ell-1$ representatives, etc.}

\item[Lower Quota.] Another approach to investigating proportionality of ABC rules is to analyze how they behave for certain structured preferences; for example, when the voters and the candidates can be divided into disjoint groups so that each group of voters approves exactly one group of candidates. Such profiles can be represented as party-list elections, and so the classic notions of proportionality for apportionment methods~\cite{BaYo82a, Puke14a} apply. For instance, it can be shown that EJR generalizes the definition of lower-quota. Informally speaking, for party-list elections lower quota requires that if a population of at least $\ell \cdot \nicefrac{n}{k}$ voters vote for a single party, then this party gets at least $\ell$ seats in the elected committee. Brill~et~al~\shortcite{bri-las-sko:c:apportionment} discussed how different ABC rules behave for such restricted preferences, and Lackner and Skowron~\shortcite{lac-sko:t:approval-thiele} proved that under certain assumptions the behavior of ABC rules uniquely extends from party-list profiles to general preferences. Our study, on the other hand, answers whether the considered rules still behave proportionally for general preferences. 
\end{description}


\section{Proportionality Degree of Phragm\'{e}n's Sequential Rule}\label{sec:phragmen}

In this section we establish the proportionality degree of the Phragm\'{e}n's Sequential Rule. However, we first provide an alternative definition of the rule that will be more convenient to work with. According to our new definition the voters gradually earn virtual money (credits) which they then use to buy committee members. Specifically, we assume that each voter earns money with the constant speed of one credit per unit time (time is continuous, not discrete). Buying a candidate costs $n$ credits, and a voter pays only for a candidate that she approves of. 
We say that a candidate $c$ is \emph{electable} if $c$ is approved by the voters who altogether have at least $n$ credits. In the first time moment when there exists an electable candidate $c$ the rule adds this candidate to the committee (ties are broken arbitrarily) and resets to zero the credits of all voters who approve $c$---intuitively, these voters pay the total amount of $n$ for adding $c$ to the committee. The rule stops when $k$ candidates are selected.

Let us now argue that the so-described process is equivalent to the Phragm\'{e}n's Sequential Rule. First, it is apparent that in the original definition of the Phragm\'{e}n's Sequential Rule we can assume that each candidate is associated with $n$ units of load instead of one. When a candidate $c$ is selected then its load is distributed so that all the voters who approve $c$ have the same total load, which is the maximum load among all the voters. The same candidate would be selected by the above described process, and each voter $v$ approving $c$ would pay for $c$ the number of credits which is equal to the difference between its current total load and the previous one (just before $c$ was selected). The formal proof of the equivalence of the two definitions is provided in \Cref{seq:phragmen_equivalence}.


Using the new definition allows us to obtain a new accurate estimation of the proportionality degree for the Phragm\'{e}n's Sequential Rule. Informally speaking, our guarantee says that the Phragm\'{e}n's Sequential Rule can be at most twice less proportional than PAV. The crucial element of the proof is the analysis of a certain potential function.
Intuitively, this function measures how unfair the committee iteratively built by the Phragm\'{e}n's Sequential Rule can be; we will show that 
\begin{inparaenum}[(i)]
\item this unfairness cumulates, and that
\item eventually the accumulated unfairness must be used to compensate the voters who got less representatives than they deserved.
\end{inparaenum}
(The most difficult part of the proof was coming up with the right potential function.)

\begin{theorem}\label{thm:phrag_guarantee}
The proportionality degree of Phragm\'{e}n Sequential Rule satisfies $d_{\mathrm{Phrag}}(\ell) \geq \frac{\ell-1}{2}$.
\end{theorem}
\begin{proof}
In the proof we will be using the alternative definition of the Phragm\'{e}n's Sequential Rule, using the concept of virtual money (credits).

Let $W$ be a committee returned by Phragm\'{e}n's Sequential Rule.
For the sake of contradiction, let us assume that there exists an $\ell$-large group of voters $V$ with $|\bigcap_{i \in V} A(i)| \geq \frac{\ell-1}{2}$ and with the average satisfaction lower than $\frac{\ell-1}{2}$. We set $n_V = |V|$. 

Observe that the rule stops after at least $k$ time units. Indeed, the total amount of credits earned by all the voters in the first $k$ time units is equal to $kn$. This allows to buy at most $k$ candidates. 

The total amount of credits collected by the voters from $V$ in the first $k$ time units is equal to~$n_V k$. The number of credits left to $V$ after the whole committee is elected is at most equal to $n$ (otherwise, these credits would have been spent earlier for buying an additional candidate approved by all the voters from $V$; such a candidate would exist as $|\bigcap_{i \in V} A(i)| \geq \frac{\ell-1}{2}$ and $\mathrm{sat}_V(A, W) < \frac{\ell-1}{2}$). Thus, the voters from $V$ spent at least $n_V k - n$ credits for buying candidates they approve.

We now move to the central argument of the proof---we will estimate how much on average a voter from $V$ pays for buying a committee member. We analyze the following potential function.

Let $p_t(i)$ denote the number of credits held by voter $i \in V$ at time $t$; further, let $p_t(\mathrm{av}) = \frac{1}{n_V}\sum_{i \in V} p_t(i)$. For a time $t$ we define the potential value $\phi_t$ as:
\begin{align*}
\phi_t = \sum_{i \in V} \big(p_t(i) - p_t(\mathrm{av}) \big)^2\text{.} 
\end{align*}
We will now analyze how the potential value changes over time. First, observe that the potential value remains unchanged when the number of credits of each voter is incremented (earning credits does not change the potential value). Next, we analyze what happens when the voters use their credits to pay for the committee members that they approve. Consider the time moment $t$ when a new committee member $c$ is selected. The voters who approve $c$ pay for her with their credits; furthermore, a voter who pays uses all her available credits. We consider these voters separately, one by one. Consider a voter $j \in V$ and assume that her number of credits decreases to 0 (since we consider the voters separately, the number of credits held by each other voter remains unchanged). In such a case the average $p_t(\mathrm{av})$ decreases by $\frac{p_t(j)}{n_V}$. We assess the change in the potential value:
\begin{align*} 
\Delta_{\phi} &= \sum_{i \in V, i \neq j} \left(p_t(i) - \left(p_t(\mathrm{av}) - \frac{p_t(j)}{n_V}\right) \right)^2 + \left(0 - \left(p_t(\mathrm{av}) - \frac{p_t(j)}{n_V}\right)\right)^2 \\
              &\qquad - \sum_{i \in V} \big(p_t(i) - p_t(\mathrm{av}) \big)^2 \\
              &= \sum_{i \in V} \left(p_t(i) - \left(p_t(\mathrm{av}) - \frac{p_t(j)}{n_V}\right) \right)^2 - \sum_{i \in V} \big(p_t(i) - p_t(\mathrm{av}) \big)^2 \\
              &\qquad + \left(p_t(\mathrm{av}) - \frac{p_t(j)}{n_V}\right)^2 - \left(p_t(j) - \left(p_t(\mathrm{av}) - \frac{p_t(j)}{n_V}\right)\right)^2 \\
              &= \sum_{i \in V} \frac{p_t(j)}{n_V} \left( 2p_t(i) - 2p_t(\mathrm{av}) + \frac{p_t(j)}{n_V} \right) - p_t(j)^2 + 2p_t(j)\left(p_t(\mathrm{av}) - \frac{p_t(j)}{n_V}\right)
\end{align*} 
Now, observe that $\sum_{i \in V} \left( 2p_t(i) - 2p_t(\mathrm{av})\right) = 0$, thus:
\begin{align*}
\Delta_{\phi} &=  \frac{p_t(j)^2}{n_V} + p_t(j)\left(2p_t(\mathrm{av}) - \frac{2p_t(j)}{n_V}- p_t(j)\right) \\
              &= p_t(j)\left(2p_t(\mathrm{av}) - \frac{p_t(j)(n_V + 1)}{n_V}\right) \text{.}
\end{align*} 
We further observe that in each time $t$ we have $p_t(\mathrm{av}) \leq \frac{n}{n_V}\leq \frac{k}{\ell}$ as otherwise the sum of credits within the group would be greater than $n$, and such credits would be earlier spent on buying a committee member who is approved by all the voters within the group. 

Let us now interpret the above calculations. Intuitively, we will argue that a voter from $V$ will on average pay at most $\frac{2n_V}{n_V+1}\cdot\frac{k}{\ell}$ for a committee member. Let us set $x_{t, j} = p_t(j) - \frac{2n_V}{n_V+1}\cdot\frac{k}{\ell}$. Then,
\begin{align*}
\Delta_{\phi} &= p_t(j)\left(2p_t(\mathrm{av}) - 2\frac{k}{\ell} - \frac{x_{t, j}(n_V + 1)}{n_V}\right) \leq -\frac{x_{t, j}(n_V + 1)}{n_V} \cdot p_t(j) \text{.}
\end{align*}
If $x_{t, j} > 0$, then $\phi$ decreases by at least $2|x_{t, j}| \cdot \frac{k}{\ell}$. Similarly, if $x_{t, j} \leq 0$, then $\phi$ increases by at most $2|x_{t, j}| \cdot \frac{k}{\ell}$. Since the potential value is always non-negative, we infer that the values of $x_{t', j'}$ for $j' \in V$ and $t'$ are on average lower than or equal to $0$. Recall that $p_t(j)$ in the definition of $x_{t, j}$ is equal to how much voter $j$ pays for the selected candidate. Thus, the voters from $V$ on average pay at most $\frac{2n_V}{n_V + 1}\cdot\frac{k}{\ell} < \frac{2k}{\ell}$ for a committee member. Consequently, the average number of committee members that the voters from $V$ approve equals at least:
\begin{align*}
\frac{n_V k - n}{\frac{2k}{\ell} \cdot n_V} = \frac{k}{\frac{2k}{\ell}} - \frac{n}{\frac{2k}{\ell} \cdot n_V} \geq \frac{\ell}{2} - \frac{n}{\frac{2k}{\ell} \cdot \frac{n\ell}{k}}  = \frac{\ell - 1}{2} \text{.}
\end{align*}
This leads to a contradiction and so completes the proof.
\end{proof}

\cref{prop:phrag_hard_instance} below upper-bounds the proportionality degree of Phragm\'{e}n's Sequential Rule. For the sake of simplicity we present the proof for the case when $k$ is divisible by $\ell$; an analogous construction holds also in the general case (with slightly worse, but asymptotically the same, bounds), but the analysis is more complex.

\begin{proposition}\label{prop:phrag_hard_instance}
The $k$-proportionality degree of Phragm\'{e}n's Sequential Rule satisfies the following inequality: For each $k, \ell \in \naturals$, $\ell < \frac{k}{2}$ and $k$ divisible by $\ell$ we have $d_{\mathrm{Phrag}}(\ell, k) \leq \frac{\ell}{2} \cdot \frac{2k - 2\ell + 2}{2k - 3\ell}$.
\end{proposition}
\begin{proof}
Let us fix two natural numbers, $\ell, k \in \naturals$; $\ell < \frac{k}{2}$ and $k$ is divisible by $\ell$. Let $x = \frac{2(k-\ell)}{\ell} - 1$ and let $L$ be an integer divisible by $\ell$, $x$, and by $k$. 
We construct a profile with $n = \frac{Lk}{\ell}$ voters and $2k + \ell$ candidates $C= \{b_1, \ldots, b_k, c_1, \ldots, c_k, d_1, \ldots, d_{\ell}\}$, as follows:
\begin{align*}
&N(b_1) = \{1, 2, \ldots, \nicefrac{L}{x}\}  \cup  \{L + 1, L + 2, \ldots, n - \nicefrac{L}{x}\}, \\
&N(b_2) = \{\nicefrac{L}{x}+1, \ldots, \nicefrac{2L}{x}\}  \cup  \{L + 1, L + 2, \ldots, n - \nicefrac{2L}{x}\}, \\
&\ldots \\
&N(b_{x}) = \{L-\nicefrac{L}{x} + 1, \ldots, L\} \cup \{L + 1, \ldots, n - L\}, \\
&N(b_{x+1}) = \{1, 2, \ldots, \nicefrac{L}{x} \} \cup \{L + 1, \ldots, n - L\}, \\
&N(b_{x+2}) = \{\nicefrac{L}{x}+1, \ldots, \nicefrac{2L}{x}\} \cup \{L + 1, \ldots, n - L\}, \\
&\ldots \\
&N(b_{k}) = \big(\text{a cyclic shift of~}\{1, 2, \ldots, \nicefrac{L}{x} \}\big) \cup \{L + 1, \ldots, n - L\} \\
&N(c_1) = \ldots = N(c_{k}) = \{n - L, n - L + 1, \ldots, n\} \\
&N(d_1) = \ldots = N(d_{\ell}) = \{1 , \ldots, L\}.
\end{align*}
In particular, candidate $b_1$ is approved by $n-L$ voters in total, $b_2$ is approved by $n-L -\nicefrac{L}{x}$ voters, and $b_3$ by $n-L -\nicefrac{2L}{x}$ voters. Each candidate from $\{b_x, \ldots, b_k\}$ is approved by the same $n - 2L$ voters from $\{L + 1, \ldots, n\}$ and by some $\nicefrac{L}{x}$ voters from $\{1, \ldots, L\}$, which cyclically shift. The voters from  $\{1, \ldots, L\}$ who approve $b_i$, $i\geq x$ are those who are right after the voters who approve $b_{i-1}$-th candidate, unless those who approve $b_{i-1}$ form the last segment, i.e., $\{L-\nicefrac{L}{x} + 1, \ldots, L\}$. In such a case, the voters from  $\{1, \ldots, L\}$ who approve $b_{i}$ are exactly $\{1, 2, \ldots, \nicefrac{L}{x}\}$.

Let $t = \frac{k}{k -\ell}$; this ensures that $t(n-L) = \frac{k}{k -\ell}\left(n - n\frac{\ell}{k}\right) = n$. In our profile Phragm\'{e}n's Sequential Rule selects $b_1$ first at time $t$. Indeed, at this time the group of voters $N(b_1)$ collects $(n - L)t = n$ credits. At time $2t$ each voter from $\{\nicefrac{L}{x}+1, \ldots, \nicefrac{2L}{x}\}$ has already $2t$ credits, and each voter from $\{L + 1, L + 2, \ldots, n - \nicefrac{2L}{x}\}$ has $t$ credits; altogether, they have $n$ credits, thus $b_2$ is selected second. By a similar reasoning, we infer that in the first $x$ steps candidates $b_1, \ldots, b_x$ will be selected by Phragm\'{e}n's Sequential Rule, and the last one of them will be selected at time $xt$. At this time the rule would also select one candidate from $\{c_1, \ldots, c_k\}$. Indeed, at time $xt$ the voters from $n - L, n - L + 1, \ldots, n$ have the following number of credits:
\begin{align*}
\frac{L}{x}\cdot t \cdot (1 + 2 + \ldots + x) = \frac{L}{x} \cdot \frac{k}{k -\ell} \cdot \frac{x(x+1)}{2} = \frac{n\ell(x+1)}{2(k - \ell)} = n \text{.}
\end{align*}

Let us now analyze what happens in the next steps. First, let us consider how the Phragm\'{e}n's Sequential Rule would behave if there were no candidates from $\{c_1, \ldots, c_k\}$.
At time $(x+1)t$ voters from $\{1, 2, \ldots, \nicefrac{L}{x}\}$ have $xt$ credits each. Similarly, each voter from $\{\nicefrac{L}{x} + 1, \ldots, \nicefrac{2L}{x}\}$ has $(x-1)t$ credits, each voter from $\{\nicefrac{2L}{x} + 1, \ldots, \nicefrac{3L}{x}\}$ has $(x-2)t$ credits, etc. The amount of credits held by the voters from  \{1, 2, \ldots, L\} altogether at time $(x+1)t$ is:
\begin{align*}
\frac{L}{x}\cdot t_1 \cdot (1 + 2 + \ldots + x) =  n \text{.}
\end{align*}
At the same time voters from $N(b_{x+1})$ have $n$ credits altogether, thus, $b_{x+1}$ can be selected next, before any candidate from $D = \{d_1, \ldots, d_k\}$ is chosen\footnote{Here we assume adversarial tie breaking, yet the construction can be strengthen so that it does not depend on the particular tie-breaking mechanism.}. Similarly, at time $(x+2)t$ each voter from $\{1, 2, \ldots, \nicefrac{L}{x}\}$ has $t$ credits, each from $\{\nicefrac{L}{x}+1, \ldots, \nicefrac{2L}{x}\}$ has $xt$ credits, each from $\{\nicefrac{2L}{x}+1, \ldots, \nicefrac{3L}{x}\}$ has $(x-1)t$ credits, etc.; altogether they have at most $n$ credits, and the voters from $N(b_{x+2})$ have exactly $n$; thus $b_{x+2}$ can be selected next. Through a similar reasoning we conclude that in the first $k$ steps Phragm\'{e}n's Sequential Rule would select candidates $\{b_1, \ldots, b_{k}\}$.

Now, consider the candidates from $\{c_1, \ldots, c_k\}$. After candidates $b_1, \ldots, b_x$ are selected, candidates from $\{c_1, \ldots, c_k\}$ are approved only by voters who do not approve other remaining candidates. Thus, their selection does not interfere with the relative order of selecting the other candidates. Further, observe that over time $(x+1)t$ the last $L$ voters collect the following number of credits:
\begin{align*} 
(x+1)tL = \frac{2(k-\ell)}{\ell} \cdot \frac{k}{k -\ell} \cdot n\frac{\ell}{k} = 2n \text{.}
\end{align*}
Thus, every $\frac{(x+1)t}{2}$ time moments the rule will select one candidate from $\{c_1, \ldots, c_k\}$. Consequently, in the first $t(k - \left\lfloor \frac{2k - 2x}{x+ 3} \right\rfloor)$ time moments, the rule will select at least $\left\lfloor \frac{2k - 2x}{x+ 3} \right\rfloor + 1$ candidates from $\{c_1, \ldots, c_k\}$. Indeed, the first candidate will be selected after the first $xt$ time moments. After the remaining $t(k - x - \lfloor \frac{2k - 2x}{x+ 3} \rfloor)$ time moments the number of candidates from $\{c_1, \ldots, c_k\}$ that will be selected is equal to:
\begin{align*} 
\frac{t(k - x - \lfloor \frac{2k - 2x}{x+ 3} \rfloor)}{\frac{(x+1)t}{2}} \geq  \frac{2(k - x - \frac{2k - 2x}{x+ 3})}{x+1} = \frac{2k - 2x}{x+ 3} \geq \left\lfloor \frac{2k - 2x}{x+ 3} \right\rfloor \text{.}
\end{align*}
Next, observe that:
\begin{align*} 
\left\lfloor \frac{2k - 2x}{x+ 3} \right\rfloor + 1 \geq \left\lfloor \frac{2k - \frac{4k - 6\ell}{\ell}}{\frac{2k - 3\ell}{\ell} + 3}\right\rfloor + 1 = \left\lfloor \frac{2k\ell - 4k + 6\ell}{2k} \right\rfloor \geq \ell - 1 \text{.}
\end{align*}
Consequently, the winning committee $W$ has at most $k - \ell + 1$ candidates from $\{b_1, \ldots, b_{k-1}, b_k\}$ and the other committee members are from $\{c_1, \ldots, c_{k-1}, c_k\}$.

The group of voters $V = \{1, \ldots, L\}$ is $\ell$-cohesive. Let us assess their average number of representatives. Observe that, except for candidates from $\{c_1, \ldots, c_{k-1}, c_k\}$, each candidate from the selected committee is approved by exactly $\frac{L}{x}$ voters of $V$. Thus:
\begin{align*}
\frac{1}{|V|}\sum_{i \in V} |W \cap A(i)| = \frac{1}{L} \cdot \frac{L}{x} \cdot \left(k - \ell + 1\right) = \frac{\ell}{2} \cdot \frac{2k - 2\ell + 2}{2k - 3\ell} \text{.}
\end{align*}
\end{proof}

Next, observe that for large enough $k$ the proportionality guarantee from \Cref{prop:phrag_hard_instance} can be arbitrarily close to $\nicefrac{\ell}{2}$. Thus, we obtain the following corollary, which shows that the guarantee from \cref{thm:phrag_guarantee} is almost tight (up to an additive constant of $\nicefrac{1}{2}$).

\begin{corollary}\label{cor:phragmen_tight}
The proportionality degree of Phragm\'{e}n's Seq. Rule satisfies $d_{\mathrm{Phrag}}(\ell) \leq \frac{\ell}{2}$.
\end{corollary}

Let us now discuss the consequences of our results from the perspective of a decision maker facing the problem of choosing the right rule. First, Phragm\'{e}n's Sequential Rule offers a considerably lower proportionality guarantee than PAV, hence the latter should be recommended whenever proportionality is the primary concern. Second, the worst-case loss of proportionality for the Phragm\'{e}n's Sequential Rule  is moderate. Thus, in some cases this loss can be compensated by other appealing properties of the rules. For example, Phragm\'{e}n's Sequential Rule satisfies committee enlargement monotonicity\footnote{Informally speaking, a rule $\calR$ satisfies committee enlargement monotonicity (sometimes, also referred to as simply committee monotonicity) if increasing the committee size from $k$ to $k+1$ results only in adding an additional candidate to the winning committees (the candidates that were selected by the rule for the committee size equal to $k$ are still selected for the committee size equal to $k+1$). For the precise definition, see~\cite{elk-fal-sko-sli:c:multiwinner-rules}.}, which makes it applicable when the goal is to compute a representative \emph{ranking} of alternatives (the recent work of Skowron~et~al.~\shortcite{proprank} discusses several domains where finding proportional rankings is critical). Further, as discussed by Janson~\shortcite{Janson16arxiv} (see, e.g., Examples 13.5 there), Phragm\'{e}n's Sequential Rule has the following appealing property:

\begin{definition}[Strong Unanimity]
An ABC rule $\calR$ satisfies \emph{strong unanimity} if for each approval-based profile $A$ such that there exists a candidate $c$ who is approved by all the voters, it holds that $\calR(A, k) = \calR(A^{-c}, k-1) \cup \{c\}$, where $A^{-c}$ denotes the profile obtained from $A$ by removing $c$ from the approval sets of all the voters.
\end{definition}

It easily follows from the definition that Phragm\'{e}n's Sequential Rule satisfies strong unanimity. Further, this property is so natural, that it is quite surprising that PAV does not satisfy it (this is an argument often raised by the critics of PAV, e.g., by Phragm\'{e}n in his original works). \Cref{thm:phrag_guarantee} shows that strong unanimity and committee enlargement monotonicity can be satisfied with only moderate loss of proportionality compared to PAV. As such, it allows to view Phragm\'{e}n's Sequential Rule as an appealing alternative for PAV (depending on how important the decision maker considers particular axioms).

\subsection{Comparing Phragm\'{e}n's Sequential Rule and Phragm\'{e}n's Maximal Rule}

Phragm\'{e}n's Maximal Rule was first introduced as an ``optimal'' alternative for its sequential counterpart. Thus, it is interesting to see that the global optimization embedded in the definition of the rule, does not translate to its better properties. In fact, it appears that Phragm\'{e}n's Maximal Rule is far from being optimal (it is not Pareto optimal and its utilitarian efficiency---the notion which we will explain in detail in \Cref{sec:utilitarian_efficiency}---is much worse than those of Phragm\'{e}n's Sequential Rule)~\cite{lac-sko:quantitative}. Below, we will complement this arguments showing that the proportionality degree of the rule is much worse than those of its sequential counterpart. Our results suggest that there is no reason to choose Phragm\'{e}n's Maximal Rule over Phragm\'{e}n's Sequential Rule.

\begin{proposition}\label{prop:max_phrag_guarantee}
The proportionality degree of Phragm\'{e}n Maximal Rule is $d_{\mathrm{max\text{-}Phrag}}(\ell) \leq 1$.
\end{proposition}
\begin{proof}
Consider the profile constructed as follows. The voters are divided into $k$ equal-size groups: $N = N_1 \cup N_2 \cup \ldots \cup N_k$.  There are $k$ candidates approved by $N$---let us call the set of these $k$ candidates $W_1$. Further, there are $k$ candidates, $W_2 = \{c_1, \ldots c_k\}$, such that for each $i \in [k]$, $N(c_i) = N_i$. Clearly, Phragm\'{e}n Maximal Rule can select $W_2$ as this committee allows to distribute the load evenly among the voters. With $W_2$ each voter will have a single representative. On the other hand, for each $\ell$ the group $N_1 \cup \ldots \cup N_{\ell}$ is $\ell$-large and has more than $\ell$ commonly approved candidates. Intuitively, for this profile $W_1$ should be chosen (in particular, it would be the single committee returned by Phragm\'{e}n's Sequential Rule).
\end{proof}

\section{Proportionality Degree of Sequential PAV}\label{sec:seqpav}
In this section we will assess the proportionality degree of Sequential PAV, with the aim of comparing Seq-PAV with Phragm\'{e}n's Sequential Rule. Our method here is quite different from the one we used in the previous section. Specifically, instead of proving a bound on the proportionality degree of Seq-PAV directly, we will show how to construct an algorithm that given a desired committee size $k$ finds in polynomial time an upper bound on the $k$-proportionality degree of the rule. We will compute these bounds for $k \leq 200$ and will argue that they are fairly accurate. Since for many natural applications of ABC rules the intended committee size is much smaller than 200, our results give a direct answer to the decision maker who is interested in estimating the proportionality degree for Seq-PAV. Yet, our result is more generic and allows a mechanism designer to make a comparison of rules having a specific committee size in mind (indeed, a committee size is usually known and fixed before the rule is used).

Designing such an algorithm is not straightforward---the main challenge lies in reducing the size of the space that needs to be searched. Indeed, even for a fixed committee size there is an infinite number of possible profiles (even if we fix the number of voters $n$ and the number candidates $m$ we still have exponentially many, namely $2^{mn}$, possible preference profiles). Thus, we will use several observations that will allow us to reformulate the problem and to compactly represent it as an instance of Linear Programming (LP).   

\subsection{The High-Level Strategy of the Algorithm}

The high-level idea of our approach can be summarized as follows. First, in~\cref{sec:seq_pav_objective} we describe a different optimization function that accurately estimates the proportionality degree of Seq-PAV  (at first this optimization function can appear loosely related to the concept of proportionality degree). This optimization function is somehow easier to work and in \Cref{sec:first_lp} we show that it can be solved by a carefully constructed linear programming (LP): this LP takes the committee size as input and returns the proportionality degree of Seq-PAV. This algorithm is exponential in the committee size but does not depend on the number of voters nor the number of candidates, thus it allows us to derive almost exact values of the proportionality degree of Sequential PAV for small values of $k$. (Up to this point, we can formally prove that the so-constructed LP returns accurate estimations of the proportionality degree up to an additive constant of 1.)

Next, in \cref{sec:relaxed_lp} we show a more complex LP which we argue is a relaxation of the first exact one. The second LP works in polynomial time in the committee size $k$, and also its formulation not depend on $n$ nor on $m$. This LP also provides formal, analytical proportionality guarantees for Seq-PAV, thus these values can be used to formally compare Seq-PAV with other rules. Indeed, we show that the so-obtained guarantees (i.e., the lower bounds on the proportionality degree) for Seq-PAV are considerably better than the proportionality degree for the Phragm\'{e}n's Sequential Rule. While these results allow us to compare Seq-PAV with the Phragm\'{e}n's Sequential rule, we still do not know how tight are the guarantees returned by our second LP.  While the formal analysis of the second LP is much harder, we observe that for small values of $k$ the two LP-s return very similar guarantees (for $k \leq 20$ they differ by at most 2\%) and we observe that the plot of the guarantees returned by our second LP for different values of $k$ flattens very quickly (e.g., the guarantees for $k = 50$ and $k = 200$ equal to $0.7085$ and $0.694$, respectively). This allows us to conjecture that the guarantees returned by our second LP are fairly accurate, at least for reasonable committee sizes.   

Finally, in \Cref{sec:seq_pav_hardness} we discuss why obtaining analytical bounds that accurately estimate the proportionality degree of Seq-PAV for all values of $k$ is a hard task. We suggest it as an interesting and challenging open question.

\subsection{Revising the Optimization Objective}\label{sec:seq_pav_objective}

Let us introduce some additional notation. Let $A$ and $k$ be an approval-based profile and a desired committee size, respectively. Let $\calR_{\seqpav}(A, k)$ be a multiset of tied winning committees returned by Sequential Proportional Approval Voting for $A$ and $k$ (here, we use a multiset because the same committee could be obtained by different tie-breaking decisions during the execution of the rule). For each winning committee $W \in \calR_{\seqpav}(A, k)$ let $c_{\last}(A, W)$ denote the candidate that Sequential Proportional Approval Voting has added to $W$ as the last committee member. Let $\Delta_{\seqpav}(A, k, W)$ denote the average, per voter, marginal increase of the PAV score due to adding $c_{\last}(A, W)$ to $W$:
\begin{align*}
\Delta_{\seqpav}(A, k, W) = \frac{1}{n}\Big( \score_{\pav}(A, W) - \score_{\pav}(A, W\setminus\{c_{\last}(A, W)\}) \Big)\text{.}
\end{align*} 
We denote the maximum possible average marginal increase as $\Delta_{\seqpav}(A, k)$:
\begin{align*}
\Delta_{\seqpav}(A, k) \; = \; \max_{\mathclap{W \in \calR_{\seqpav}(A, k)}} \; \Delta_{\seqpav}(A, k, W) \text{.}
\end{align*} 
Further, we define $\Delta_{\seqpav}(k)$ as:
\begin{align*}
\Delta_{\seqpav}(k) = \sup_{A \in \calA} \Delta_{\seqpav}(A, k) \text{.}
\end{align*}
Our first lemma shows a close relation between the proportionality guarantee of Sequential Proportional Approval Voting and the value $\Delta_{\seqpav}(k)$. This is very useful, as the definition of $\Delta_{\seqpav}(k)$ is not based on $\ell$-large cohesive groups, and thus it is much easier to handle. 

\begin{lemma}\label{lem:seq_pav_guarantee_transformed}
The $k$-proportionality degree of Sequential PAV satisfies:
\begin{align*}
d_{\seqpav}(\ell, k) \geq \ell \cdot \frac{1}{k \cdot \Delta_{\seqpav}(k)} - 1 \text{.}
\end{align*}
For each $k \in \naturals_{+}$ the proportionality degree of Sequential PAV satisfies:
\begin{align*}
d_{\seqpav}(\ell) \leq \ell \cdot \frac{1}{k \cdot \Delta_{\seqpav}(k)} \text{.}
\end{align*} 
\end{lemma}

The proof is rather technical, thus we redelegate it to the appendix.

As an immediate corollary of \Cref{lem:seq_pav_guarantee_transformed} we obtain an almost tight estimation of the proportionality guarantee of Sequential Proportional Approval Voting.

\begin{corollary}
The proportionality guarantee of Sequential Proportional Approval Voting is $d_{\seqpav}(\ell) = \inf_{k}\frac{\ell}{k \cdot \Delta_{\seqpav}(k)} - \epsilon$ for some $\epsilon \in [0, 1]$.
\end{corollary}

In the remaining part of this section we will focus on assessing the expression $h_{\seqpav}(k) = k \cdot \Delta_{\seqpav}(k)$, for different values of the parameter $k$. According to \Cref{lem:seq_pav_guarantee_transformed} this expression will give us the lower bound on the $k$-proportionality degree and an upper bound on the proportionality degree of Sequential PAV.

\subsection{An Exact Linear Program for Assessing Proportionality Degree}\label{sec:first_lp}

For each committee size $k$ we can compute $h_{\seqpav}(k)$ by solving an appropriately constructed linear program; such an LP for each committee size $k$ finds a profile $A$ for which $k \cdot \Delta_{\seqpav}(A, k)$ is maximal. Designing such an LP however requires some care, since ideally its size should not depend on the number of voters nor the number of candidates. 
Our first observation is that we can consider only the profiles with $k$ candidates (all of which form a winning committee). Indeed, if we take a profile $A$ that maximizes $k \cdot \Delta_{\seqpav}(A, k)$, then we can remove from $A$ all candidates which are not members of the winning committee. After this removal we obtain a profile that still witnesses the maximality of $k \cdot \Delta_{\seqpav}(A, k)$. Thus, we will assume that in the profile that we look for there are $k$ candidates, and we will represent them as integers from $[k]$. Further, without loss of generality we will assume that the candidates are added to the winning committee in the order of their corresponding numbers: candidate $1$ is added first, candidate $2$ is added second, etc. 

Second, we observe that we do not need to represent each voter, we can rather cluster the voters into the groups having the same approval sets and for each possible approval set $T \subseteq [k]$ we are only interested in the proportion of the voters who have this approval set. Such a proportion will be denoted by the variable $x_T$. Now we can use a Linear Program formulation given by S{\'a}nchez-Fern{\'a}ndez~et~al.~\shortcite{pjr17} to compute $h_{\seqpav}(k)$.\footnote{S{\'a}nchez-Fern{\'a}ndez~et~al.~\shortcite{pjr17} used a very similar LP to compute the smallest committee size $k$ for which Seq-PAV violates JR. Our so-far analysis shows that a very similar LP can be used to obtain much stronger guarantees that lead to roughly tight estimates of the proportionality degree of Seq-PAV.}
In \Cref{fig:ilpFormulation1} we give the LP; for a logical expression $E$, we set $\mathds{1}[E]$ to be 1 if $E$ is true, and 0 otherwise.\footnote{Note that the LP is guaranteed to return a rational solution, which can be converted in an exact profile.}   

\begin{figure}[th!]
\begin{align*}
    &\text{maximize } \quad k \cdot \sum_{T \in S([k])} x_T \cdot \frac{\mathds{1}[k \in T]}{|T|}  \ & \\
    &\text{subject to: } \ & \\  
    & \text{(a)}: \sum_{T \in S([k])} x_T = 1 \ &                                                                       \ \\
    & \text{(b)}: \sum_{T \in S([k])} \frac{x_T \cdot \mathds{1}[i \in T]}{|T \cap [i]|} \geq \sum_{T \in S([k])} \frac{x_T \cdot \mathds{1}[j \in T]}{|T \cap ([i-1] \cup \{j\})|}  \ &   \quad, i,j \in [k]; j > i  \\
 & \text{(c)}: x_T \geq 0  \ &  \quad, T \in S([k]) 
\end{align*}
\caption{Linear programming (LP) formulation for computing $h_{\seqpav}(k) = k \cdot \Delta_{\seqpav}(k)$.}
\label{fig:ilpFormulation1}
\end{figure}

The expression that we maximize is exactly $k \cdot \Delta_{\seqpav}(A, k)$. Indeed, $\frac{\mathds{1}[k \in T]}{|T|}$ is the marginal increase of the PAV score coming from a voter who approves $T$, as a result of adding candidate $k$ to committee $[k-1]$. The constraint (a) ensures that the proportions of clustered voters sum up to 1. The constraint (b) ensures that in the $i$-th step of Sequential PAV, the marginal increase of the PAV score due to adding $i$ to committee $[i-1]$ is at least as large as due to adding $j > i$ to $[i-1]$. These constraints ensure that the candidates are indeed added in the order $1, 2, \ldots, k$.

We computed the above program for $k \leq 20$; the resulting lower bounds for the $k$-proportionality degree of Sequential PAV are given in \Cref{tab:seq_pav_lower_bound}. 

\begin{table}[tb!]
\begin{center}
   \begin{minipage}{.21\linewidth}
   \centering
   \begin{tabular*}{\linewidth}{@{\extracolsep{\fill}}@{}rl@{}}
    \toprule
    $k$ & lower-bound \\ 
    \midrule
    1 & 1.0 $\ell$ \\ 
    2 & 1.0 $\ell$ \\ 
    3 & 0.8888 $\ell$ \\ 
    4 & 0.8571 $\ell$ \\  
    5 & 0.8372 $\ell$ \\ 
    \bottomrule
   \end{tabular*}
   \end{minipage}%
   \hspace{0.5cm}
   \begin{minipage}{.21\linewidth}
   \centering
   \begin{tabular*}{\linewidth}{@{\extracolsep{\fill}}@{}rl@{}}
    \toprule
    $k$ & lower-bound \\ 
    \midrule  
    6 & 0.8169 $\ell$ \\ 
    7 & 0.8064 $\ell$ \\ 
    8 & 0.7979 $\ell$ \\ 
    9 & 0.7888 $\ell$ \\ 
    10 & 0.7825 $\ell$ \\ 
    \bottomrule
   \end{tabular*}
   \end{minipage} 
 \hspace{0.5cm}
   \begin{minipage}{.21\linewidth}
   \centering
   \begin{tabular*}{\linewidth}{@{\extracolsep{\fill}}@{}rl@{}}
    \toprule
    $k$ & lower-bound \\ 
    \midrule
    11 & 0.7773 $\ell$ \\ 
    12 & 0.7719 $\ell$ \\ 
    13 & 0.7684 $\ell$ \\ 
    14 & 0.7647 $\ell$ \\ 
    15 & 0.7616 $\ell$ \\ 
    \bottomrule
   \end{tabular*}
   \end{minipage} 
 \hspace{0.5cm}
   \begin{minipage}{.21\linewidth}
   \centering
   \begin{tabular*}{\linewidth}{@{\extracolsep{\fill}}@{}rl@{}}
    \toprule
    $k$ & lower-bound \\ 
    \midrule
    16 & 0.7589 $\ell$ \\ 
    17 & 0.7563 $\ell$ \\ 
    18 & 0.7540 $\ell$ \\ 
    19 & 0.7522 $\ell$ \\ 
    20 & 0.7503 $\ell$ \\ 
    \bottomrule
   \end{tabular*}
   \end{minipage} 
\end{center}
\caption{A lower-bound on the $k$-proportionality degree of Sequential PAV.}
\label{tab:seq_pav_lower_bound}

\begin{center}
   \begin{minipage}{.21\linewidth}
   \centering
   \begin{tabular*}{\linewidth}{@{\extracolsep{\fill}}@{}rl@{}}
    \toprule
    $k$ & lower-bound \\ 
    \midrule
    1 & 1.0 $\ell$ \\ 
    2 & 1.0 $\ell$ \\ 
    3 & 0.8888 $\ell$ \\ 
    4 & 0.8461 $\ell$ \\ 
    5 & 0.8307 $\ell$ \\ 
    \bottomrule 
   \end{tabular*}
   \end{minipage}%
   \hspace{0.5cm}
   \begin{minipage}{.21\linewidth}
   \centering
   \begin{tabular*}{\linewidth}{@{\extracolsep{\fill}}@{}rl@{}}
    \toprule
    $k$ & lower-bound \\ 
    \midrule
    6 & 0.8131 $\ell$ \\ 
    7 & 0.7952 $\ell$ \\ 
    8 & 0.7871 $\ell$ \\ 
    9 & 0.7771 $\ell$ \\ 
    10 & 0.7705 $\ell$ \\ 
    \bottomrule
   \end{tabular*}
   \end{minipage} 
 \hspace{0.5cm}
   \begin{minipage}{.21\linewidth}
   \centering
   \begin{tabular*}{\linewidth}{@{\extracolsep{\fill}}@{}rl@{}}
    \toprule
    $k$ & lower-bound \\ 
    \midrule
    11 & 0.7643 $\ell$ \\ 
    12 & 0.7594 $\ell$ \\ 
    13 & 0.7548 $\ell$ \\ 
    14 & 0.7512 $\ell$ \\ 
    15 & 0.7476 $\ell$ \\ 
    \bottomrule
   \end{tabular*}
   \end{minipage} 
 \hspace{0.5cm}
   \begin{minipage}{.21\linewidth}
   \centering
   \begin{tabular*}{\linewidth}{@{\extracolsep{\fill}}@{}rl@{}}
    \toprule
    $k$ & lower-bound \\ 
    \midrule
    16 & 0.7441 $\ell$ \\ 
    17 & 0.7416 $\ell$ \\ 
    18 & 0.7396 $\ell$ \\ 
    19 & 0.7371 $\ell$ \\ 
    20 & 0.7348 $\ell$ \\ 
    \bottomrule
   \end{tabular*}
   \end{minipage} 
\end{center}
\caption{A lower-bound on the $k$-proportionality degree of Sequential PAV, as computed by the Linear Program from \Cref{fig:ilpFormulation2}.}
\label{tab:seq_pav_lower_bound2}

\end{table}

Let us compare the values from \Cref{tab:seq_pav_lower_bound} with the bounds for the Phragm\'{e}n's Sequential Rule from \Cref{thm:phrag_guarantee} and \Cref{prop:phrag_hard_instance}. Since the bounds for the Phragm\'{e}n's Sequential Rule are more accurate for smaller groups of voters (when $k$ is significantly larger than $\ell$), let us consider the committee size $k = 10$, and an $\ell$-large group consisting of 10\% of the voters  (thus, $\ell = 1$). In this case, \Cref{prop:phrag_hard_instance} says that the $k$-proportionality degree of the Phragm\'{e}n's Sequential Rule cannot be higher than $0.625 \ell$, while the $k$-proportionality degree of Sequential PAV equals at least $0.7825 \ell$. This proves that Sequential PAV has better proportionality guarantees than the Phragm\'{e}n's Sequential Rule for $k = 10$; a similar comparison can be performed for $10 \leq k \leq 20$.     

\subsection{A Relaxed LP for Assessing Proportionality Degree}\label{sec:relaxed_lp}

Unfortunately, the LP from \Cref{fig:ilpFormulation1} is exponential in $k$, so we were able to provide accurate bounds only for $k \leq 20$. In the remaining part of this section we will describe another LP, that runs in polynomial time in $k$. The new LP does not compute the exact value of $h_{\seqpav}(k) = k \cdot \Delta_{\seqpav}(k)$, but rather an upper bound for $h_{\seqpav}(k)$. However, we will show that for small values of $k$ ($k \leq 20$) this lower bound is very accurate. The new LP will allow us to compute a lower-bound on the $k$-proportionality degree of Sequential PAV for $k \leq 200$; we will see that for larger values of $k$ this bound is still considerably better than the upper-bounds for the Phragm\'{e}n's Sequential Rule.

In the new LP we use the same observations as before: (i) we consider profiles with $k$ candidates only, and (ii) we cluster the voters into groups, and for each group we have a variable denoting how large it is as a fraction of all the voters. Specifically, we have the following variables. For each $i \in [k]$ we have a variable $a_i$ that denotes the fraction of voters who approve exactly $i$ from all candidates in total. For $i \in [k], j, p \in [k]_{0}$ and $p \leq i, j$, we have variable $b_{i, j, p}$ that denotes the fraction of the voters who approve $i$ candidates in total, and that after the $j$-th step of Sequential PAV approve exactly $p$ from the already selected committee members. Finally, for $i, j, p \in [k]$ and $p \leq i, j$, we have variable $c_{i, j, p}$ that denotes the fraction of voters who approve $i$ candidates in total, and that as a result of the $j$-th step of Sequential PAV the number of their representatives increased from $p-1$ to~$p$. Finally, for $j \in [k]$ we have a variable $d_j$ that denotes the increase of the PAV score after the $j$-th step of Sequential PAV, multiplied by $\nicefrac{1}{n}$.   
The new LP is given in \Cref{fig:ilpFormulation2}.

\begin{figure}[t!]
\begin{align*}
    &\text{maximize } \quad k \cdot d_k  \ & \\
    &\text{subject to: } \ & \\  
    & \text{(a1)}: a_i \geq 0  \ &  \quad, i \in [k] \\ 
    & \text{(a2)}: \sum_{i \in [k]} a_i = 1 \ &     \\
    & \text{(b1)}: b_{i, 0, 0} = b_{i, k, i} = a_i \ & \quad, i \in [k]  \\
    & \text{(b2)}: b_{i, k, p} = 0 \ & \quad, p < i  \\
    & \text{(c1)}: c_{i, j, p} \geq 0 \ & \quad, i, j, p \in [k], p\leq i, j \\
    & \text{(c2)}: c_{i, j, p} \leq b_{i, j-1, p-1} \ & \quad, i, j, p \in [k], p\leq i,j \\
    & \text{(d1)}: b_{i, j, j} = c_{i, j, j} \ & \quad, i, j \in [k], j\leq i \\
    & \text{(d2)}: b_{i, j, i} = b_{i, j-1, i} + c_{i, j, i} \ & \quad, i, j \in [k], i \leq j \\
    & \text{(d3)}: b_{i, j, 0} = b_{i, j-1, 0} - c_{i, j, 1} \ & \quad, i, j \in [k] \\
    & \text{(d4)}: b_{i, j, p} = b_{i, j-1, p} - c_{i, j, p+1} + c_{i, j, p} \ & \quad, i, j, p \in [k], p \leq i-1, j-1 \\
    & \text{(e1)}: d_{j} = \sum_{i \in [k]}\sum_{p \leq i, j} \frac{c_{i, j, p}}{p}  \ & \quad, j \in [k] \\
    & \text{(e2)}: d_{j} \geq \frac{1}{k - j + 1}\cdot \sum_{i \in [k]}\sum_{p \leq i, j-1} (i - p) \frac{b_{i, j-1, p}}{p+1}  \ & \quad, j \in [k] \\
\end{align*}
\caption{Linear programming (LP) formulation for computing a upper bound for $h_{\seqpav}(k)$.}
\label{fig:ilpFormulation2}
\end{figure}

Let us now explain the constraints in the LP. Constraints (a1)-(a2) and (c1) are intuitive. Constraint (b1)-(b2) say that before Sequential PAV starts the committee is empty (and so every voter in each group has no representatives) and that after Sequential PAV finishes, all the candidates are selected (and so every voter in each group has as many representatives, as the number of approved candidates). Constraint (c2) says that the voters whose number of representatives increased from $p-1$ to $p$ in the $j$-th step, must have had $p-1$ representatives after the $(j-1)$-th step of Sequential PAV. Constraints (d1)-(d4) give the natural relations between $b$- and $c$- variables. Constraint (e1) encodes the definition of the $d$-variables. Constraint (e2) is the most tricky, and it is an incarnation of the pigeonhole principle. Consider the expression on the right-hand side of (e2). Consider the voters who approve $i$ candidates in total, and who before the $j$-th step of Sequential PAV, approve $p$ already elected members. The fraction of such voters is $b_{i, j-1, p}$. Each such a voter approves $i - p$ not yet elected candidates. Adding each such a candidate would increase the score that each such a voter assigns to a committee by $\frac{1}{p+1}$. Thus, the sum in the right-hand side of (e2) gives the total, over all not yet selected candidates, increase of the PAV score due to adding such candidates to the current committee; the whole sum is normalized by $\nicefrac{1}{n}$. There are $k - j + 1$ not yet selected candidates. Thus, by the pigeonhole principle, adding at least one of such candidate increases the normalized PAV score by at least the value which is in the right-hand side of (e2). Since Sequential PAV selects the candidate which increases the (normalized) PAV score most, we get Constraint (e2).  

The LP from \Cref{fig:ilpFormulation2} gives a set of constraints that each preference profile must satisfy. However, some solutions to the LP do not encode any valid profile; intuitively, the LP finds a solution in a larger space than the space of all valid preference profiles. This is why it only provides an upper bound for $h_{\seqpav}(k)$, in contrast to the LP from \Cref{fig:ilpFormulation1}, which computes its exact value. Nevertheless, as we will argue, the value computed by the LP from \Cref{fig:ilpFormulation2} is accurate enough to formulate interesting conclusions. Indeed, \Cref{tab:seq_pav_lower_bound2} provides the results of the computation of the new LP for different values of $k$. One can observe that these values are very close to those from \Cref{tab:seq_pav_lower_bound}, computed by the exact LP from \Cref{fig:ilpFormulation1}. For example, for $k = 20$, these values are $0.7085$ and $0.694$, respectively. Further, in \Cref{fig:seq_pav_bounds_all}, we depict these lower-bounds for $k \leq 200$, and we observe that they are considerably better than the bounds for the Phragm\'{e}n's Sequential Rule. Additionally, by analyzing \Cref{tab:seq_pav_lower_bound} and \Cref{tab:seq_pav_lower_bound2} we observe that the differences between the values computed by the two LP-s increase with $k$; this allows us to conjecture that the actual $k$-proportionality degree of Sequential PAV can be even better than the bounds presented in~\Cref{fig:seq_pav_bounds_all} (note that the bounds from~\Cref{fig:seq_pav_bounds_all} were computed with the relaxed LP, thus these bounds are firm but possibly not optimal as if they were computed by the exact LP).

\begin{figure}[tb]
    \begin{center}
      \includegraphics[scale=0.75]{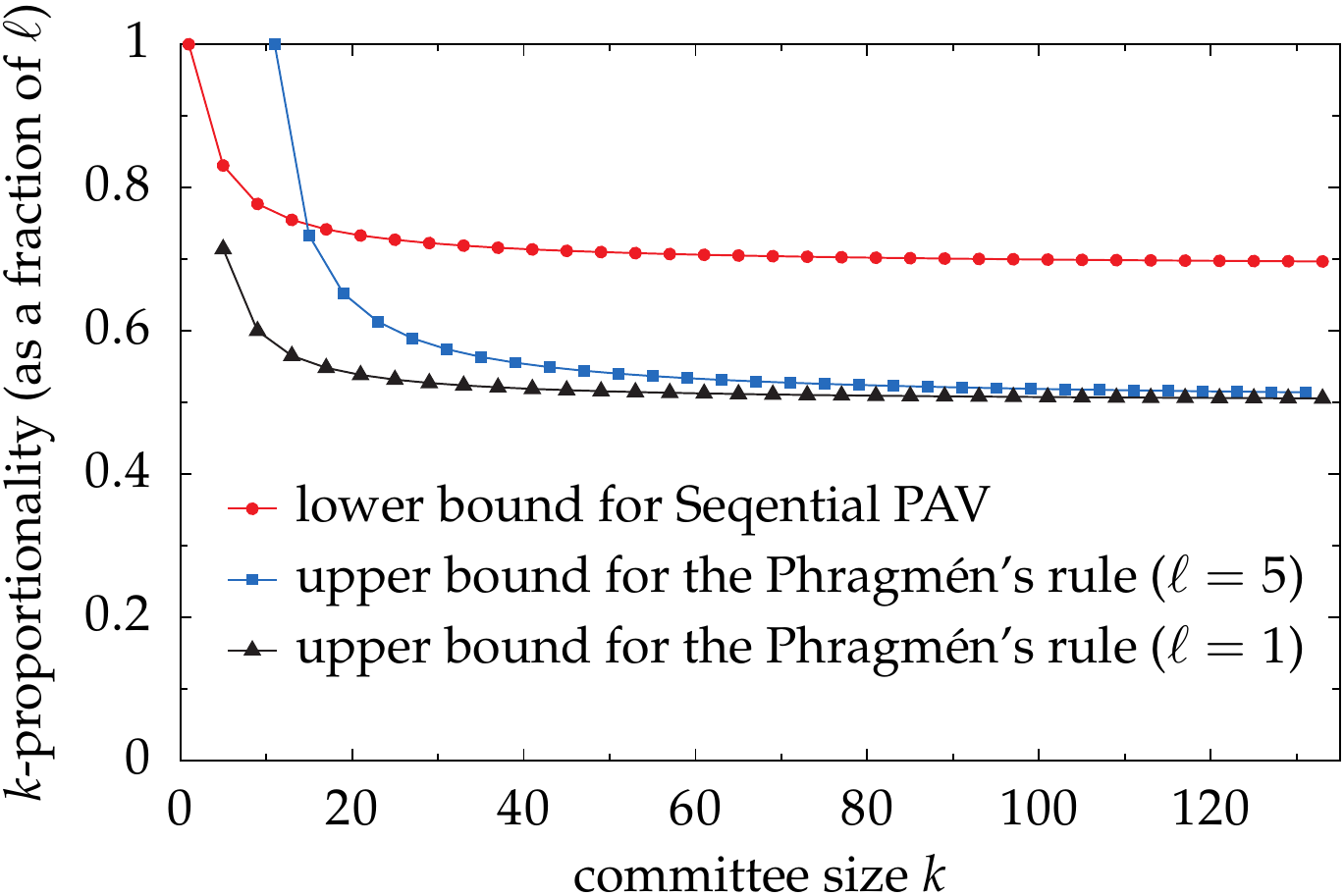}
    \end{center}
    \caption{The lower bounds on the $k$-proportionality degree of Sequential PAV as computed by the LP from \Cref{fig:ilpFormulation2}. The plot compares these lower-bounds with the bounds for the Phragm\'{e}n's Sequential Rule, as given in \Cref{prop:phrag_hard_instance}.}
    \label{fig:seq_pav_bounds_all}
\end{figure}

Summarizing, Sequential PAV should be preferred over PAV and over the Phragm\'{e}n's Sequential Rule, when proportionality and committee enlargement monotonicity are the primary requirements. On the other hand, Sequential PAV does not satisfy strong unanimity---if a decision maker considers this a serious flaw, then the Phragm\'{e}n's Sequential Rule is a good alternative.

\section{Proportionality Degree of Thiele Methods}

In this section we establish the bounds on the proportionality degrees for a large subclass of Thiele methods. We first obtain a generic result that covers the whole class of rules, and then we show how the general result can be used to assess the proportionality degrees for concrete rules in the class; for the selected concrete rules, we will show that our estimations are accurate. In the next section we will analyze the same subclass from the perspective of utilitarian efficiency. Our results will show that there is a whole spectrum of rules that implement different tradeoffs between proportionality and utilitarian efficiency and will allow us to accurately quantify these tradeoffs. 

\begin{theorem}\label{thm:general_bounds}
Let $\lambda\colon \reals \to \reals_{+}$ be a non-increasing, convex function, and let $g\colon \naturals \times \naturals \to \reals$ be a function satisfying $g(\ell, k) \leq k$, and the following equality: 
\begin{align}\label{eq:proportionality_condition}
(k - g(\ell, k))\lambda(1 + g(\ell, k)) = \frac{k - \ell}{\ell} \cdot \max_{x \in [k]} x \lambda(x) \qquad \text{for each $\ell, k \in [m], \ell \leq k$.}
\end{align}
Then, the $\lambda$-Thiele rule has the $k$-proportionality guarantee of $g$.
\end{theorem}
\begin{proof}
Let $W$ be a committee winning according to the $\lambda$-Thiele rule.
Consider an $\ell$-large group of voters $V$ with $|\bigcap_{i \in V}A(i)| \geq g(\ell, k)$, and for the sake of contradiction, assume that $\frac{1}{|V|} \sum_{i \in V} |W \cap A(i)| < g(\ell, k)$. From the pigeonhole principle we infer that there exists a candidate $c \notin W$ who is approved by each voter from $V$ (and, possibly, by some voters outside of $V$). Let us now estimate the change of the $\lambda$-score due to replacing a committee member $c' \in W$ with $c$:
\begin{align*}
\Delta(c, c') &\geq \sum_{i \in V\colon c' \notin A(i)} \lambda(|A(i) \cap W| + 1) - \sum_{i \notin V\colon c' \in A(i)} \lambda(|A(i) \cap W|) \text{.}
\end{align*} 
Let us now assess the following sum:
\begin{align*}
\sum_{c' \in W}\Delta(c, c') &\geq \sum_{i \in V, c' \in W \setminus A(i)} \lambda(|A(i) \cap W| + 1) - \sum_{i \notin V, c' \in W \cap A(i)} \lambda(|A(i) \cap W|) \\
                             &\geq \sum_{i \in V} (k - |A(i) \cap W|) \cdot \lambda(|A(i) \cap W| + 1) - \sum_{i \notin V} |A(i) \cap W| \cdot \lambda(|A(i) \cap W|) \\
                             &\geq \sum_{i \in V} (k - |A(i) \cap W|) \cdot \lambda(|A(i) \cap W| + 1) - (n - |V|)\cdot \max_{x \in [k]} x \lambda(x) \text{.}
\end{align*}
From the Jensen's inequality, we get that:
\begin{align*}
&\sum_{i \in V} \frac{(k - |A(i) \cap W|)}{\sum_{j \in V} (k - |A(j) \cap W|)} \cdot \lambda(|A(i) \cap W| + 1) \geq \lambda\left(\sum_{i \in V} \frac{(k - |A(i) \cap W|)\cdot (|A(i) \cap W| + 1)}{\sum_{j \in V} (k - |A(j) \cap W|)}\right) \\
&\qquad\qquad = \lambda\left(1 + \frac{\sum_{i \in V} (k - |A(i) \cap W|)\cdot |A(i) \cap W|}{\sum_{i \in V} (k - |A(i) \cap W|)}\right) \\
&\qquad\qquad \text{(from Chebyshev's sum inequality, and by monotonicity of $\lambda$)} \\
&\qquad\qquad \geq \lambda\left(1 + \frac{\sum_{i \in V} (k - |A(i) \cap W|)\cdot \sum_{i \in V}|A(i) \cap W|}{|V|\sum_{i \in V} (k - |A(i) \cap W|)}\right) \\
&\qquad\qquad = \lambda\left(1 + \frac{\sum_{i \in V}|A(i) \cap W|}{|V|}\right) \geq \lambda(1 + g(\ell, k)) \text{.}
\end{align*}
Consequently, we get that:
\begin{align*}
\sum_{c' \in W}\Delta(c, c') &\geq \sum_{i \in V} (k - |A(i) \cap W|) \lambda(1 + g(\ell, k)) - (n - |V|)\cdot \max_{x \in [k]} x \lambda(x) \\
                             &> |V|(k - g(\ell, k))\lambda(1 + g(\ell, k)) - (n - |V|)\cdot \max_{x \in [k]} x \lambda(x) \text{.}
\end{align*}
Since for each $c' \in W$ we have that $\Delta(c, c') \leq 0$, it holds that:
\begin{align*}
(n - |V|)\cdot \max_{x \in [k]} x \lambda(x) > |V|(k - g(\ell, k))\lambda(1 + g(\ell, k)) \text{,}
\end{align*}
from which it follows that:
\begin{align*}
\frac{k - \ell}{\ell} \cdot \max_{x \in [k]} x \lambda(x) > (k - g(\ell, k))\lambda(1 + g(\ell, k)) \text{,}
\end{align*}
a contradiction. This completes the proof.
\end{proof}

Let us first argue that Equality~\eqref{eq:proportionality_condition} has always a unique solution, i.e., that for each positive, non-increasing, convex function $\lambda$ \Cref{thm:general_bounds} provides a single $k$-proportionality guarantee $g$. Observe that $(k - g(\ell, k))$ is a decreasing function of $g(\ell, k)$, and that $\lambda(1 + g(\ell, k))$ is a non-increasing function of $g(\ell, k)$. As a result, for $g(\ell, k) \leq k$ the left-hand side of~\eqref{eq:proportionality_condition} is a decreasing function of $g(\ell, k)$. For $g(\ell, k) = k$ the left hand side of~\eqref{eq:proportionality_condition} is lower or equal to the right-hand side; for $g(\ell, k) = -k^2$ we have the opposite, thus for each $\ell$ and $k$ there exists a unique value $g(\ell, k)$ that satisfies Equality~\eqref{eq:proportionality_condition}.

\Cref{prop:general_bound_tight}, below provides a condition that is necessary for $g$ to be a $k$-proportionality degree for $\lambda$. Thus, by finding a function $g(\ell, k)$ that violates this condition we obtain an upper-bound on the proportionality degree of the rule. We will get the most accurate upper-bound if we treat the condition there as an equality and solve it for $g$.
The equations in \Cref{thm:general_bounds}~and~\Cref{prop:general_bound_tight} differ slightly, so we cannot formally prove that our estimation given in \Cref{thm:general_bounds} is always tight. Yet, we solved the two equations for a number of representative Thiele rules (in particular, for $\lambda(i) = \nicefrac{1}{i}$, $\lambda(i) = \nicefrac{1}{\sqrt{i}}$, $\lambda(i) = \left(\nicefrac{1}{i}\right)^{\nicefrac{2}{3}}$, and $\lambda(i) = \nicefrac{1}{i^2}$) and verified that the obtained lower and upper bounds are very close. 

\begin{proposition}\label{prop:general_bound_tight}
Let $\lambda\colon \reals \to \reals$ be a non-increasing, convex function. The $k$-proportionality guarantee $g$ of the $\lambda$-Thiele rule must satisfy the following inequality: 
\begin{align*}
(k - g(\ell, k))\lambda(g(\ell, k)) \geq \frac{k - \ell}{\ell} \cdot \max_{x \in [k]} x \lambda(x+1) \qquad \text{for each $\ell, k \in [m], \ell \leq k$.}
\end{align*}
\end{proposition}

\begin{proof}
Let $y = \argmax_{i \in [k]} i \lambda(i)$, and for the sake of contradiction, let us assume that for some $\ell, k \in [m]$, and for some $\epsilon > 0$ it holds that:
\begin{align*}
(k - g(\ell, k))\lambda(g(\ell, k)) + \epsilon < \frac{k - \ell}{\ell} \cdot y \lambda(y+1) \text{.}
\end{align*}
To simplify the notation we set $x = g(\ell, k)$. Clearly, $x \leq k$. We rewrite the above inequality:
\begin{align*}
(k - x)\lambda(x) + \epsilon < \frac{k - \ell}{\ell} \cdot y \lambda(y+1) \text{.}
\end{align*}

We will construct an instance of an election witnessing that $g$ cannot be a proportionality guarantee of the $\lambda$-Thiele rule. 
Let $C = B \cup D$ be the set of candidates, where $B=\{b_1, b_2, \ldots, b_k\}$ and $D = \{d_1, d_2, \ldots, d_k\}$.
We divide the voters into two disjoint groups, $V$ and $V'$, such that $|V| = n\cdot\frac{\ell}{k}$ and $|V'| = n - |V|$. Without loss of generality, let us assume that $n$ is divisible by $k^2$. Each candidate $b \in B$ is approved by all the voters from $V$. Candidate $d_1$ is approved by the first $|V|\frac{x}{k}$ voters from $V$ and by the first $|V'|\frac{y}{k}$ voters from $V'$. 
The approval set of candidate $d_i$ is constructed by taking the $|V|\frac{x}{k}$ voters from $V$ that appear right after the voters from $V \cap N(d_{i-1})$ and by taking $|V'|\frac{y}{k}$ voters from $V'$ that appear right after the voters from $V' \cap N(d_{i-1})$, taking the cyclic shift if necessary. Finally, from the approval set of the last candidate $d_k$ we remove one arbitrary voter.

We will show that $D$ is an optimal committee for this instance. For the sake of contradiction assume that an optimal committee $W$ contains $z \geq 1$ candidates from $B$. Then, the voters from $V$ and $V'$ approve on average $\frac{x}{k}(k - z) + z$ and $\frac{y}{k}(k - z)$ committee members, respectively. Since $\lambda$ is convex, given a constraint on the total number of committee members approved, the score of a committee is maximized when the voters approve roughly the same number of candidates in $W$. Thus, we can assume that each voter from $V$ approves $\left\lfloor \frac{x}{k}(k - z)\right\rfloor + z$ or $\left\lceil \frac{x}{k}(k - z)\right\rceil + z$ members of $W$ and, analogously, that each voter from $V'$ approves $\left\lfloor \frac{y}{k}(k - z)\right\rfloor$ or $\left\lceil \frac{y}{k}(k - z)\right\rceil$ members of $W$. Consequently, if we replace one candidate from $B$ with a candidate from $D$ in $W$, then the score of the committee will change by:
\begin{align*}   
\Delta &\geq |V'|\frac{y}{k}\lambda\left(\left\lceil \frac{y}{k}(k - z)\right\rceil + 1\right) - \left(|V| - |V|\frac{x}{k} + 1 \right)\lambda\left(\left\lfloor \frac{x}{k}(k - z)\right\rfloor + z\right) \\
       &\geq |V'|\frac{y}{k}\lambda\left(\left\lceil \frac{y}{k}\cdot k\right\rceil + 1\right) - \left(|V| - |V|\frac{x}{k} + 1 \right)\lambda\left(\left\lfloor \frac{x}{k}\cdot k\right\rfloor\right) \\
       &= |V'|\frac{y}{k}\lambda(y + 1) - \left(|V| - |V|\frac{x}{k} + 1\right)\lambda(x) \\
       &= \left(n - n\frac{\ell}{k}\right)\frac{y}{k}\lambda(y + 1) - n\frac{\ell}{k}\left(1 - \frac{x}{k} + \frac{k}{n\ell}\right)\lambda(x) \text{.}
\end{align*}
Since $\Delta \leq 0$, we get that:
\begin{align*}   
0 &\geq (k - \ell)y\lambda(y + 1) - \ell(k - x)\cdot\lambda(x) - \frac{k^2}{n}\lambda(x) \text{,}
\end{align*}
which is equivalent to:
\begin{align*}   
 (k - x)\cdot\lambda(x) + \frac{k^2}{n\ell}\lambda(x) \geq \frac{k - \ell}{\ell}y\lambda(y + 1) \text{.}
\end{align*}
By taking $n$ large enough, so that $\frac{k^2}{n\ell}\lambda(x) \leq \epsilon$, we get a contradiction. Consequently, we get that $D$ forms a winning committee. Now, observe that $V$ is $\ell$-large, and that $|\bigcap_{i \in V}A(i)| \geq x$. The average number of representatives of $V$ is however, lower than $x$. Thus, $g$ is not a proportionality guarantee for the $\lambda$-Thiele method. This gives a contradiction, and completes the proof. 
\end{proof}

Let us give an example application of \Cref{thm:general_bounds}, by considering $\lambda_{\sqrt{\pav}}(i) = \nicefrac{1}{\sqrt{i}}$. In this case, the equality from \Cref{thm:general_bounds} becomes:
\begin{align}
\frac{k - \ell}{\ell} \cdot \sqrt{k} = \frac{k - g(\ell, k)}{\sqrt{1 + g(\ell, k)}} \text{.}
\label{eq:relation_for_sqrt_pav}
\end{align}
Let us calculate a proportionality guarantee $g_{\sqrt{\pav}}$ for $\lambda_{\sqrt{\pav}}$-Thiele method by solving~\eqref{eq:relation_for_sqrt_pav}, and picking the solution for $g(\ell, k)$ that is no greater than $k$. We solved the equality analytically, yet the formulas for the bounds have two lines in displayed equations. Thus, we decided that it is more informative if we present these guarantees in a form of plots rather than closed formulas. We plot $g_{\sqrt{\pav}}$ and the guarantees for three other rules in~\Cref{fig:sqrt_pav_bounds_all} (the plots for upper and lower bounds are in fact indistinguishable for these rules, thus our estimation for them are really tight). Interestingly, we can see that the loss of proportionality in comparison to PAV---especially for reasonably large cohesive groups of voters---is moderate. This can justify using rules such as the $\lambda_{\sqrt{\pav}}$-Thiele method in some contexts. For instance, as we will see in the next section, such rules guarantee much better utilitarian efficiency.

\begin{figure}[tb]
    \begin{center}
      \includegraphics[scale=0.75]{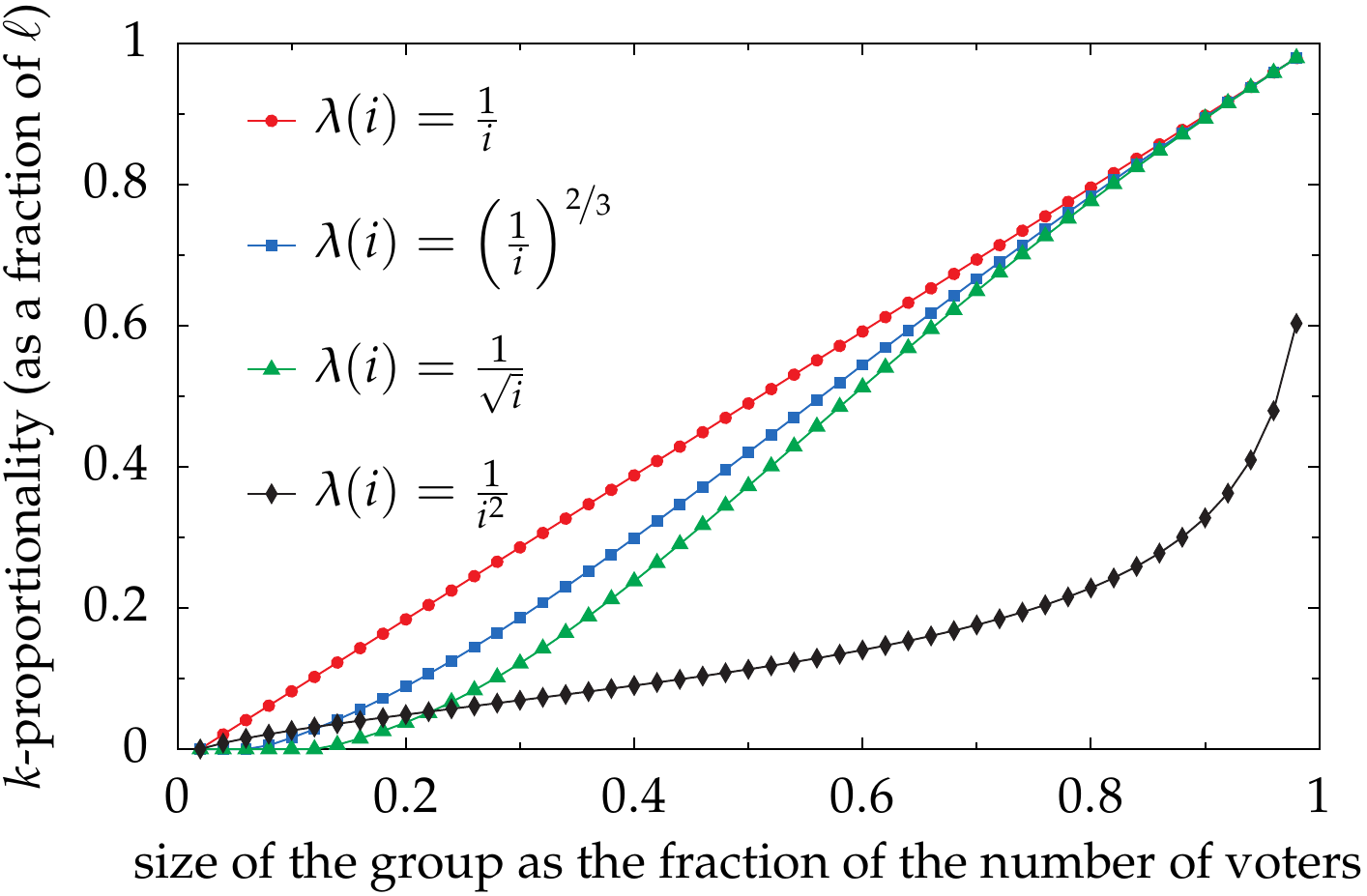}
    \end{center}
    \caption{The lower bounds on the $k$-proportionality degree of selected $\lambda$-Thiele rules.}
    \label{fig:sqrt_pav_bounds_all}
\end{figure}

\section{The Tradeoff Between Proportionality and Utilitarian Efficiency}\label{sec:utilitarian_efficiency}

In the previous sections we assessed the proportionality degree for a number of committee election rules. Here, our goal is estimate the utilitarian efficiency of the considered rules. Together with our previous results, this will allow us to express the studied rules as implementing certain tradeoffs between proportionality and utilitarian efficiency. 

We define the utilitarian efficiency of a rule as a lower bound on the ratio between the total number of approvals that the candidates from the winning committee get from all the voters, to the total number of approvals that any committee can possibly get. Formally:

\begin{definition}[Utilitarian Efficiency]
For a given committee size $k$ an ABC rule $\calR$ has a utilitarian efficiency guarantee of $\alpha$ if for each winning committee $W \in \calR(A, k)$ we have:
\begin{align*}
\sum_{i \in N}|A(i) \cap W| \geq \alpha \cdot \max_{W_{\opt} \in S_k(C)} \sum_{i \in N}|A(i) \cap W_{\opt}| \text{.}
\end{align*}
\end{definition}

The next theorem gives a generic tool for calculating the utilitarian efficiency for a large subclass of Thiele methods. In particular, \Cref{thm:utilitarian_efficiency} generalizes some known results from the literature---an asymptotically tight utilitarian efficiency guarantee is already known for PAV and for the $p$-geometric rule~\cite{lac-sko:quantitative}.  

\begin{theorem}\label{thm:utilitarian_efficiency}
Let $\lambda\colon \naturals \to \reals$ be a non-increasing, convex function. For a committee size $k$, the $\lambda$-Thiele rule has the utilitarian efficiency guarantee of $\frac{\alpha}{1 + \alpha} > \alpha - \alpha^2$, for $\alpha$ defined as:
\begin{align*}
  \alpha\lambda(1) = \lambda\left(1 + k\alpha \right) \text{.}
\end{align*}
\end{theorem}
\begin{proof}
Let $W_{\calR}$ and $W_{\opt}$ denote, respectively, the committee selected by rule $\calR$, and the committee maximizing the total utility of the voters. For each voter $i \in N$, let $w_i = |A(i) \cap W_{\calR}|$. Without loss of generality, we can assume that $W_{\calR} \neq W_{\opt}$. Let $c$ be the candidate that is approved by most voters among the candidates from $W_{\opt} \setminus W_{\calR}$, and let $n_c = |N(c)|$. For each candidate $c' \in W_{\calR}$ let $\Delta(c, c')$ denote the change of the total $\lambda$-score due to replacing $c'$ with $c$ in committee $W_{\calR}$. Since $W_{\calR}$ is optimal according to $\calR$, we have that $\Delta(c, c') \leq 0$ for each $c' \in W_{\calR}$. Thus:
\begin{align*}
0 &\geq \sum_{c' \in W_{\calR}}\Delta(c, c') = \sum_{c' \in W_{\calR}} \left(\sum_{i \in N(c)\setminus N(c')} \lambda(w_i + 1) - \sum_{i \in N(c')\setminus N(c)} \lambda(w_i) \right) \\
  &= \sum_{i \in N(c)} (k - w_i) \lambda(w_i + 1) - \sum_{i \notin N(c)} w_i \lambda(w_i) \\
  &\geq \sum_{i \in N(c)} k \lambda(w_i + 1) - \sum_{i \in N} w_i \lambda(\max(w_i, 1)) \\
  &\geq \sum_{i \in N(c)} k \lambda(w_i + 1) - \lambda(1) \sum_{i \in N} w_i \text{.}
\end{align*}
Thus:
\begin{align*}
  \lambda(1) \sum_{i \in N} w_i  \geq \sum_{i \in N(c)} k \lambda(w_i + 1) \text{.}
\end{align*}
Since $\lambda$ is convex, by the Jensen's inequality we get that:
\begin{align*}
  \lambda(1) \sum_{i \in N} w_i \geq k n_c \lambda\left(1 + \frac{1}{n_c}\sum_{i \in N(c)}w_i \right) \text{.}
\end{align*}

Now, let us consider two cases. If $\left(\sum_{i \in N} w_i\right) \geq \alpha k n_c$, then the ratio of total utilities of $W_{\opt}$ and $W_{\calR}$ is at most equal to:
\begin{align*}
\tau = \frac{\sum_{i \in N} w_i + kn_c}{\sum_{i \in N} w_i} \leq 1 + \frac{1}{\alpha} = \frac{1 + \alpha}{\alpha} \text{.}
\end{align*}
Thus, $\nicefrac{1}{\tau}$ equals at least $\nicefrac{\alpha}{1 + \alpha}$. Otherwise, i.e., when $\left(\sum_{i \in N} w_i\right) < \alpha k n_c$, we have that:
\begin{align*}
  \alpha\lambda(1) > \lambda\left(1 + \frac{1}{n_c}\sum_{i \in N(c)}w_i \right) \geq \lambda\left(1 + \frac{1}{n_c}\sum_{i \in N}w_i \right) \geq \lambda\left(1 + k\alpha \right) \text{,}
\end{align*}
a contradiction. Thus, we get that only the first case is possible, and so the ratio of total utilities of $W_{\calR}$ and $W_{\opt}$ is at least $\nicefrac{\alpha}{1 + \alpha}$.
\end{proof}

Finally, the subsequent proposition provides an upper bound on utilitarian efficiency guarantees of Thiele rules. In particular, these upper bounds confirm that for many rules the guarantees given in \Cref{thm:utilitarian_efficiency} are asymptotically tight up to a multiplicative factor of 2 (as we will show later on, usually we have $\alpha = o(1)$, which gives the asymptotic tightness).\footnote{Similarly as for the proportionality---since the equations in \Cref{thm:utilitarian_efficiency}~and~\Cref{prop:utilitarian_efficiency_tight} differ slightly, we cannot claim that our estimation is asymptotically tight for \emph{every} Thiele method. Yet, solving the two equations for a number of representative rules shows the lower and the upper bounds are asymptotically the same.}

\begin{proposition}\label{prop:utilitarian_efficiency_tight}
Let $\lambda\colon \naturals \to \reals$ be a decreasing, convex function. Let $\alpha$ satisfy $\alpha\lambda(1) = \lambda\left( k\alpha \right)$.
For a committee size $k$, the utilitarian efficiency guarantee of the $\lambda$-Thiele rule is below $2\alpha - \alpha^2$.
\end{proposition}
\begin{proof}
We construct an instance of election with a set of $2k$ candidates $C = B \cup D$, where $B = \{b_1, \ldots, b_k\}$, and $D = \{d_1, \ldots, d_k\}$. The voters are divided into two groups, $V$ and $N \setminus V$, such that $|V|k\alpha = n - |V|$. Each voter from $V$ approves all the candidates from $B$. The voters from $N \setminus V$ are divided into $k$ equal-sized groups; each group approves one candidate from $D$ so that the approval sets of any two groups do not overlap. First, we will show that the elected committee $W$ for this instance has at most $\alpha k$ candidates from $B$. Indeed, if this were not the case, then by replacing in $W$ one candidate from $B$ with a candidate from $D$ we would change its score by:
\begin{align*}
\Delta > \frac{n - |V|}{k}\cdot\lambda(1) - |V|\lambda(k\alpha) = |V|\alpha\lambda(1) - |V|\lambda(k\alpha) \text{.}
\end{align*} 
Since $\Delta \leq 0$, we get that $\lambda(k\alpha) > \alpha\lambda(1)$, a contradiction. Thus, $W$ has at most $k\alpha$ candidates from $B$. Now, we assess the ratio $\tau$ of total utilities of $W$ and $B$:
\begin{align*}
\tau \leq \frac{\alpha k \cdot |V| + (k - \alpha k)\cdot \frac{n-|V|}{k}}{|V|k} = \alpha + \frac{(n-|V|)(1 - \alpha)}{|V|k} = \alpha + \alpha(1 - \alpha) = 2\alpha - \alpha^2 \text{.}
\end{align*}  
\end{proof}

The equation used in the statement of \Cref{prop:utilitarian_efficiency_tight}, $\alpha\lambda(1) = \lambda\left( k\alpha \right)$, has a unique solution: the left-hand side (LHS) is increasing, right-hand side  (RHS) is decreasing,  $\text{LHS} < \text{RHS}$ at $\alpha = \nicefrac{1}{k}$, and $\text{LHS} > \text{RHS}$ at $\alpha=1$.

E.g., from \Cref{thm:utilitarian_efficiency} and \Cref{prop:utilitarian_efficiency_tight} it follows that the utilitarian efficiency of $\lambda$-Thiele methods for $\lambda(i) = \frac{1}{i}$, $\lambda(i) = \frac{1}{\sqrt{i}}$ and $\lambda(i) = \frac{1}{i^2}$ is, respectively, $\Theta\left(\frac{1}{\sqrt{k}}\right)$, $\Theta\left(\frac{1}{\sqrt[3]{k}}\right)$, and $\Theta\left(\frac{1}{k^{\nicefrac{2}{3}}}\right)$.  

\section{Conclusion}

This paper quantifies the level of proportionality and the utilitarian efficiency for a number of multiwinner rules. We provide general tools that allow to estimate the proportionality degree and the utilitarian efficiency for a large subclass of voting rules. Our results show, in particular, a spectrum of rules that implement various tradeoffs between proportionality and utilitarian efficiency. 

For specific rules our conclusions can be summarized as follows. Phragm\'{e}n's Sequential Rule is roughly half as proportional as Proportional Approval Voting. Sequential PAV lies, in terms of proportionality, between Phragm\'{e}n's Sequential Rule and PAV. Further, some rules such as the $\lambda$-Thiele rule for $\lambda(i) = \nicefrac{1}{\sqrt{i}}$ offer a significantly better utilitarian efficiency than PAV, at the cost of a moderate loss of proportionality.

We consider the following open questions particularly interesting and important: can we find a rule that combines the virtues of Phragm\'{e}n's Sequential Rule and of PAV? Such a rule should in particular 
\begin{inparaenum}[(i)]
\item satisfy Pareto efficiency (which PAV satisfies and Phragm\'{e}n's Sequential Rule violates~\cite{lac-sko:quantitative}),
\item satisfy strong unanimity (which Phragm\'{e}n's Sequential Rule satisfies and PAV violates), and 
\item have a high proportionality degree. 
\end{inparaenum}
It is tempting to suggest a rule that first takes all unanimously approved candidates and complements the committee by running PAV, yet such rule looks a bit ad hoc, and its definition is specifically tailored for strong unanimity.\footnote{For example, for profiles where there exist candidates who are approved by ``almost'' all the voters, the rule would behave just like PAV. This is not intended as strong unanimity is only an example of a more complex set of requirements one would impose on an election rule.} Thus, finding stronger properties that better capture the idea exposed by strong unanimity is essential.

Second, can we use Phragm\'{e}n's Sequential Rule to compare any two committees? The definition of the rule allows only for finding winning committees. Comparing is, however, sometimes essential, e.g., when there are some external constraints put on the committee and when the goal is to find the best possible committee subject to these constraints. While Phragm\'{e}n proposed a few other (seemingly similar) rules based on global optimization goals (e.g., Phragm\'{e}n's Maximal Rule), these rules offer much worse proportionality and utilitarian efficiency, so the question is still open.

\bibliography{main}

\clearpage
\appendix
\section{Proofs Omitted From the Main Text}

\subsection{Proof of \Cref{lem:seq_pav_guarantee_transformed}}

\begin{replemma}{lem:seq_pav_guarantee_transformed}
The $k$-proportionality degree of Sequential PAV satisfies:
\begin{align*}
d_{\seqpav}(\ell, k) \geq \ell \cdot \frac{1}{k \cdot \Delta_{\seqpav}(k)} - 1 \text{.}
\end{align*}
For each $k \in \naturals_{+}$ the proportionality degree of Sequential PAV satisfies:
\begin{align*}
d_{\seqpav}(\ell) \leq \ell \cdot \frac{1}{k \cdot \Delta_{\seqpav}(k)} \text{.}
\end{align*} 
\end{replemma}
\begin{proof}
First, we prove that the $k$-proportionality guarantee of Sequential Proportional Approval Voting satisfies $d_{\seqpav}(\ell, k) \geq \ell \cdot \frac{1}{k \cdot \Delta_{\seqpav}(k)}-1$.
For the sake of contradiction let us assume that this is not the case, and that there exists an approval-based profile $A$ with $n$ voters, a committee size $k$, an $\ell$-large group of voters $V$ with $|\bigcap_{i \in V} A(i)| \geq \ell \cdot \frac{1}{k \cdot \Delta_{\seqpav}(k)}-1$, and a winning committee $W \in \calR_{\seqpav}(A, k)$ such that:
\begin{align*}
\frac{1}{|V|}\sum_{i \in V} |W \cap A(i)| < \ell \cdot \frac{1}{k \cdot \Delta_{\seqpav}(k)} - 1 \text{.}
\end{align*}


By the pigeonhole principle, there exists a candidate $c \notin W$ who is approved by all voters from $V$. 
Let us now estimate by how much adding $c$ to $W$ increases the PAV score. Fixing the average number of representatives that the voters from $V$ have in $W$, it is straightforward to check that the increase would be the smallest when each voter in $V$ has roughly the same number of representatives. Thus, adding $c$ to $W$ would increase the PAV score by more than
\begin{align*}
\frac{|V|}{\ell \cdot \frac{1}{k \cdot \Delta_{\seqpav}(k)} - 1 + 1} = |V| \cdot \frac{k}{\ell} \cdot \Delta_{\seqpav}(k) \geq n \Delta_{\seqpav}(k) \text{.}
\end{align*}
Clearly, adding $c$ to any subset of $W$ would result in at least the same increase of the PAV score. Since Sequential Proportional Approval Voting always selects a  candidate that increases the PAV score of the committee most, and since $c$ was not selected, we infer that $\Delta_{\seqpav}(A, k, W) > \frac{1}{n} \cdot n \Delta_{\seqpav}(k) = \Delta_{\seqpav}(k)$, a contradiction.

Second, let us fix $\epsilon > 0$ and $k \in \naturals$. We will construct an approval-based profile $A_{\hard}$ that witnesses that $d_{\seqpav}(\ell) \leq \ell \cdot \frac{1}{k \cdot \Delta_{\seqpav}(k)} + \epsilon$. Let $A$ be an approval-based profile and $W$ be a size-$k$ committee winning in $A$ such that $\Delta_{\seqpav}(A, k, W) > \Delta_{\seqpav}(k) - \epsilon_1$, for some $\epsilon_1$; the value of $\epsilon_1$ will become clear later on. Let us fix an integer $L$---intuitively, $L$ is a large number; the exact value of $L$ will also become clear from the further part of the proof. Let $n = |A|$. We construct $A_{\hard}$ by appending $L$ independent copies of $A$ (we clone both the voters and the candidates) and $y = \frac{Ln\ell}{Lk-\ell}$ voters who all approve some $\ell$ candidates, not approved by any voter from any copy of $A$---let us denote this set of $y$ voters by $V$. Let $n' = |A_{\hard}|$; clearly $n' = Ln + \frac{Ln\ell}{Lk-\ell} = \frac{L^2nk}{Lk-\ell}$. We set the required committee size to $k' = Lk$. Observe that $V$ is $\ell$-cohesive. Indeed, all voters from $V$ approve common $\ell$ candidates; further the relative size of $V$, $\nicefrac{y}{n'}$, is equal to
\begin{align*}
\frac{y}{n'} = \frac{Ln\ell}{Lk-\ell} \cdot \frac{Lk-\ell}{L^2nk} = \frac{\ell}{Lk} = \frac{\ell}{k'} \text{.}
\end{align*} 
Now, we show that the voters from $V$ have on average less than $\frac{\ell}{k \cdot \Delta_{\seqpav}(k)} + \epsilon$ representatives in some winning committee for $A_{\hard}$. Towards a contradiction, let us assume that this is not the case. Since the voters in $V$ are identical, this means that each such a voter has more than $\frac{\ell}{k \cdot \Delta_{\seqpav}(k)} + \epsilon$ representatives in each winning committee. Thus, for any winning committee, when the last representative of the voters from $V$ was added, the PAV-score of the committee increased by at most:
\begin{align*}
\Delta = y \cdot \frac{k \cdot \Delta_{\seqpav}(k)}{\ell + \epsilon k \cdot \Delta_{\seqpav}(k)} 
\end{align*}
The above expression can be written as  
\begin{align*}
\Delta = y \cdot \left(\frac{k \cdot \Delta_{\seqpav}(k)}{\ell} - \epsilon_2\right) \text{,}
\end{align*}
where $\epsilon_2$ is some parameter dependent on $\epsilon$ and $k$. Further, observe that:
\begin{align*}
\lim_{L \to \infty} \Delta &= \lim_{L \to \infty} y \cdot \left(\frac{k \cdot \Delta_{\seqpav}(k)}{\ell} - \epsilon_2\right) = \lim_{L \to \infty} \frac{Ln\ell}{Lk-\ell} \cdot \left(\frac{k \cdot \Delta_{\seqpav}(k)}{\ell} - \epsilon_2\right)\\
&= \frac{n\ell}{k} \cdot \left(\frac{k \cdot \Delta_{\seqpav}(k)}{\ell} - \epsilon_2\right) = n \Delta_{\seqpav}(k) - \epsilon_2 \cdot \frac{n\ell}{k} \text{.}
\end{align*}
Thus, there exists large enough $L$ and small enough $\epsilon_1$ such that:
\begin{align*}
\Delta < n \Delta_{\seqpav}(k) - n \epsilon_1 < n \Delta_{\seqpav}(A, k, W) \text{.}
\end{align*}
These are exactly the values of $L$ and $\epsilon_1$ that we use in our construction. In other words, the increase of the PAV score due to adding the last representative of $V$ is lower than the increase of the PAV score due to adding the last committee member to $W$. Thus, clearly there exists a winning committee that consists of some copies of the candidates from $W$ and less than $\frac{\ell}{k \cdot \Delta_{\seqpav}(k)} + \epsilon$ candidates from those approved by the voters from $V$, a contradiction. This completes the proof.
\end{proof}

\section{Equivalence of the Two Definitions of the Phragm\'{e}n's Rule}\label{seq:phragmen_equivalence}

In this section we prove that the money-based procedure described in the first part of \Cref{sec:phragmen} is equivalent to the Phragm\'{e}n's Sequential Rule.

As we already indicated, we can assume that the original definition of the Phragm\'{e}n's rule each candidate is associated with $n$ units of load instead of one. Recall that by $\ell_i(j)$ we denote the total load assigned to voter $i$ after the $j$-th iteration of the Phragm\'{e}n's rule. Let $\ell_{\max}(j)$ denote the total load of the voter who carries the maximal load after the $j$-th iteration of the rule, that is $\ell_{\max}(j) = \max_{i \in N} \ell_{i}(j)$. For each $j \in [k] \cup \{0\}$ we will prove by induction the following two statements:
\begin{enumerate}
\item The $j$-th candidate added by the Phragm\'{e}n's rule to the winning committee is the same as the $j$-th candidate added to the winning committee by the money-based procedure.
\item The value $\ell_{\max}(j) - \ell_i(j)$ is equal to the number of credits that voter $i$ is left with after the $j$-th candidate is added to the winning committee.
\end{enumerate}
Clearly, our induction hypothesis is true for $j = 0$, since at that point each of the two rules had not yet added any candidate to the committee. Assume that the induction hypothesis is true for $j = x$. We will show that it also holds for $j = x+1$. Let $t_j$ denote the time moment in the money-based procedure that corresponds to adding the $j$-th candidate to the committee. For each $\Delta$ in time $t_x + \Delta$ voter $i$ has the total amount of $\ell_{\max}(x) + \Delta - \ell_i(x)$ money (this follows by the inductive assumption, and by the fact that the voters earn money with the constant speed of one credit per time unit). Thus, the money-based procedure adds candidate $c$ the $(j+1)$-th committee member in time $t_x + \Delta$  if:
\begin{align*}
\sum_{i \in N(c)} \left(\ell_{\max}(x) + \Delta - \ell_i(x)\right) = n
\end{align*}
and if for each other candidate $c'$ we have 
\begin{align*}
\sum_{i \in N(c')} \left(\ell_{\max}(x) + \Delta - \ell_i(x)\right) < n
\end{align*}
Thus, it finds the smallest possible value $v = \ell_{\max}(x) + \Delta$ such that there exists a candidate $c$ for whom:
\begin{align*}
\sum_{i \in N(c)} \left(v - \ell_i(x)\right) = n
\end{align*}

On the other hand, the Phragm\'{e}n's rule aims at minimizing $\ell_{\max}(x+1)$. Thus, $\ell_{\max}(x+1)$ will be the minimal value $v'$ for which there exists a candidate $c$ and a load distribution $\{\delta_{i}\}_{i \in N}$ such that:
\begin{align*}
\sum_{i \in N(c)} \left(\delta_{i} + \ell_i(x)\right) = v', \quad \text{and} \quad \sum_{i \in N(c)} \delta_i = n \text{.}
\end{align*} 
These equalities can be reformulated as:
\begin{align*}
\sum_{i \in N(c)} \left(v' - \ell_i(x)\right) = \sum_{i \in N(c)} \delta_{i} = n \text{.}
\end{align*} 
Thus, we see that the minimization of $v'$ corresponds to the minimization of $v$, and so $\ell_{\max}(x+1) = v$. As a result the two rules will pick the same candidate $c$ to the committee, proving the first condition of our inductive hypothesis. Further, the money-based procedure will add this candidate after $\Delta$ time units from adding the previous one, where:
\begin{align*}
\Delta = v - \ell_{\max}(x) = \ell_{\max}(x+1) - \ell_{\max}(x) \text{.}
\end{align*}
Thus, for each $i \notin N(c)$ we will have that voter $i$ is left with the following amount of money:
\begin{align*}
\ell_{\max}(x) - \ell_i(x) + \Delta = \ell_{\max}(x+1) - \ell_i(x) = \ell_{\max}(x+1) - \ell_i(x+1)\text{.}
\end{align*}
Similarly, if $i \in N(c)$, then the number of credits left is equal to zero, but also $\ell_{\max}(x+1) = \ell_i(x+1)$; this proves the second condition of our inductive hypothesis, and so it completes the proof.

\section{Hardness of Estimating the Proportionality Degree of Seq-PAV}\label{sec:seq_pav_hardness}

Let us informally argue that obtaining a good estimation of the proportionality degree for Seq-PAV requires considerably different types of insights that are used for estimating the proportionality degree of PAV. 

Recall that in order to define PAV in \Cref{sec:rules_desc} we first formulated a scoring function $\sc_{\lambda_{\pav}} \colon S(C) \to \reals$ that assigns a numeric value to each committee; then, PAV selects a size-$k$ committee $W$ that maximizes $\sc_{\lambda_{\pav}}(W)$. The proof of \Cref{thm:pav_prop_degree}~\cite{AEHLSS18} (a tight estimation of the proportionality degree for PAV) crucially relies on one particular property: that the marginal contributions of all committee members sum up to at most $n$, i.e., that for each committee $W$:
\begin{align}\label{eq:marginal_contrib}
\left(\sc_{\lambda_{\pav}}(W) - \sc_{\lambda_{\pav}}(W\setminus \{c\}) \right) \leq n \text{.}
\end{align}

In this section we ask whether this, and other natural properties of $\sc_{\lambda_{\pav}}$ are sufficient to prove accurate estimations for the proportionality degree of Sequential PAV. In particular, we ask the following question: Assume we are given a function $f\colon S(C) \to \reals$ that is non-negative, monotonic, and that satisfies the marginal contributions property (\Cref{eq:marginal_contrib}); we ask whether the sequential algorithm for maximizing $f$ has comparably good proportionality degree to Seq-PAV. If this would be the case, it would suggest that one could focus on the aforementioned three properties of Seq-PAV in the attempt to obtaining tight proportionality guarantees.

We first constructed an LP where we represent the function $f$ through a large collection of variables, specifying the values of $f$ for all committees of size equal to at most 12. 
The three properties (being non-negative, monotonic, and the marginal contributions property) can be easily encoded in the LP. Then, we used a similar technique to the one described in~\cref{sec:seq_pav_objective,sec:first_lp} to assess the proportionality degree of the sequential algorithm for maximizing $f$. The values computed by our LP are given in \Cref{tab:seq_f_lower_bound}. These values are much lower than the ones obtained by the LPs from \Cref{sec:first_lp} and from \Cref{sec:relaxed_lp} (cf., \Cref{tab:seq_pav_lower_bound} and \Cref{tab:seq_pav_lower_bound2}). This shows that the aforementioned three properties are not sufficient to derive a good estimation of the proportionality degree for Seq-PAV.

\begin{table}[tbh]
\begin{center}
   \begin{minipage}{.21\linewidth}
   \centering
   \begin{tabular*}{\linewidth}{@{\extracolsep{\fill}}@{}rl@{}}
    \toprule
    $k$ & lower-bound \\ 
    \midrule
    1 & 1.0 $\ell$ \\ 
    2 & 1.0 $\ell$ \\ 
    3 & 0.6666 $\ell$ \\ 
    \bottomrule
   \end{tabular*}
   \end{minipage}%
   \hspace{0.5cm}
   \begin{minipage}{.21\linewidth}
   \centering
   \begin{tabular*}{\linewidth}{@{\extracolsep{\fill}}@{}rl@{}}
    \toprule
    $k$ & lower-bound \\ 
    \midrule  
    4 & 0.5 $\ell$ \\  
    5 & 0.4 $\ell$ \\ 
    6 & 0.3333 $\ell$ \\ 
    \bottomrule
   \end{tabular*}
   \end{minipage} 
 \hspace{0.5cm}
   \begin{minipage}{.21\linewidth}
   \centering
   \begin{tabular*}{\linewidth}{@{\extracolsep{\fill}}@{}rl@{}}
    \toprule
    $k$ & lower-bound \\ 
    \midrule
    7 & 0.2857 $\ell$ \\ 
    8 & 0.25 $\ell$ \\ 
    9 & 0.2222 $\ell$ \\ 
    \bottomrule
   \end{tabular*}
   \end{minipage} 
 \hspace{0.5cm}
   \begin{minipage}{.21\linewidth}
   \centering
   \begin{tabular*}{\linewidth}{@{\extracolsep{\fill}}@{}rl@{}}
    \toprule
    $k$ & lower-bound \\ 
    \midrule
    10 & 0.2 $\ell$ \\ 
    11 & 0.1818 $\ell$ \\ 
    12 & 0.1666 $\ell$ \\ 
    \bottomrule
   \end{tabular*}
   \end{minipage} 
\end{center}
\caption{The estimation of the $k$-proportionality degree of the sequential algorithm for maximizing $f$, under the assumptions that $f$ is non-negative, monotonic, and that it satisfies the marginal contributions property.}
\label{tab:seq_f_lower_bound}

\begin{center}
   \begin{minipage}{.21\linewidth}
   \centering
   \begin{tabular*}{\linewidth}{@{\extracolsep{\fill}}@{}rl@{}}
    \toprule
    $k$ & lower-bound \\ 
    \midrule
    1 & 1.0 $\ell$ \\ 
    2 & 1.0 $\ell$ \\ 
    3 & 0.8888 $\ell$ \\ 
    \bottomrule
   \end{tabular*}
   \end{minipage}%
   \hspace{0.5cm}
   \begin{minipage}{.21\linewidth}
   \centering
   \begin{tabular*}{\linewidth}{@{\extracolsep{\fill}}@{}rl@{}}
    \toprule
    $k$ & lower-bound \\ 
    \midrule  
    4 & 0.8141 $\ell$ \\  
    5 & 0.8372 $\ell$ \\ 
    6 & 0.7888 $\ell$ \\ 
    \bottomrule
   \end{tabular*}
   \end{minipage} 
 \hspace{0.5cm}
   \begin{minipage}{.21\linewidth}
   \centering
   \begin{tabular*}{\linewidth}{@{\extracolsep{\fill}}@{}rl@{}}
    \toprule
    $k$ & lower-bound \\ 
    \midrule
    7 & 0.7667 $\ell$ \\ 
    8 & 0.7492 $\ell$ \\ 
    9 & 0.7358 $\ell$ \\ 
    \bottomrule
   \end{tabular*}
   \end{minipage} 
 \hspace{0.5cm}
   \begin{minipage}{.21\linewidth}
   \centering
   \begin{tabular*}{\linewidth}{@{\extracolsep{\fill}}@{}rl@{}}
    \toprule
    $k$ & lower-bound \\ 
    \midrule
    10 & 0.7246 $\ell$ \\ 
    11 & 0.7150 $\ell$ \\ 
    12 & 0.7066 $\ell$ \\ 
    \bottomrule
   \end{tabular*}
   \end{minipage} 
\end{center}
\caption{The estimation of the $k$-proportionality degree of the sequential algorithm for maximizing $f$, under the assumptions that $f$ is non-negative, monotonic, submodular, and that it satisfies the marginal contributions property.}
\label{tab:seq_f_lower_bound2}
\end{table}

Next, we encoded one additional property of $f$: we assumed it is submodular. The values computed with this additional assumption are provided in \Cref{tab:seq_f_lower_bound2}. Here, the estimations of the proportionality degree are much better, yet still considerably worse than those derived by the relaxed LP from \Cref{sec:relaxed_lp}. This again shows that in order to obtain a tight estimation of the proportionality degree of Seq-PAV one needs some new insights and to explore other properties of $\sc_{\lambda_{\pav}}(W)$ (the ``standard'' properties and the one that was used for deriving the proportionality degree of PAV are not sufficient).

\end{document}